\documentclass[preprint]{imsart}

\usepackage{graphicx}

\usepackage{amsmath}
\usepackage{amsthm}
\usepackage{amsfonts}
\usepackage{amssymb}

\theoremstyle{plain}
\newtheorem{thm}{Theorem}[section]
\newtheorem{ass}[thm]{Assumption}
\newtheorem{alg}[thm]{Algorithm}
\newtheorem{prop}[thm]{Proposition}
\newtheorem{lemma}[thm]{Lemma}
\newtheorem{cor}[thm]{Corollary}

\theoremstyle{definition}
\newtheorem{example}[thm]{Example}
\newtheorem{defi}[thm]{Definition}

\theoremstyle{remark}
\newtheorem{remark}[thm]{Remark}

\def\lancuch{(X_n)_{n\geq 0}}
\def\lancucht{(\tilde{X}_n)_{n\geq 0}}
\def\lancuchs{(S_n)_{n\geq 0}}
\def\stany{\mathcal{X}}
\def\borel{\mathcal{B}(\stany)}

\def\1c{\mathbb{I}_C(x)}
\def\d{\textrm{d}}

\def\n0{n_{0}}

\def\Pr{\mathbb{P}}

\def\d{{\rm d}}

\def\Y{\mathcal{Y}}
\def\half{{1 \over 2}}

\begin{document}

\begin{frontmatter}

\title{Adaptive Gibbs samplers}
\runtitle{Adaptive Gibbs samplers}

\begin{aug}
\author{\fnms{Krzysztof} \snm{{\L}atuszy\'{n}ski}
\ead[label=e1]{latuch@gmail.com}}
\and
\author{\fnms{Jeffrey S.} \snm{Rosenthal}\corref{}
\ead[label=e3]{jeff@math.toronto.edu}
}

\runauthor{K. {\L}atuszy\'nski and J.S.Rosenthal}

\affiliation{University of Warwick and University of Toronto}

\address{K. {\L}atuszy\'nski\\ Department of Statistics\\
University of Warwick\\ CV4 7AL, Coventry, UK \\
\printead{e1}
}

\address{J. S. Rosenthal\\Department of Statistics\\
University of Toronto\\ Toronto, Ontario, Canada,
M5S 3G3\\
\printead{e3}}
\end{aug}

\medskip
\centerline{(January, 2010)}

\begin{abstract}
We consider various versions of adaptive Gibbs and
Metropolis-within-Gibbs samplers, which update their
selection probabilities (and perhaps also their proposal
distributions) on the fly during a run, by learning as
they go in an attempt to optimise the algorithm.  We
present a cautionary example of how even a simple-seeming
adaptive Gibbs sampler may fail to converge.  We then
present various positive results guaranteeing convergence
of adaptive Gibbs samplers under certain conditions.
\end{abstract}



\end{frontmatter}

\section{Introduction} \label{sec_intro}

Markov chain Monte Carlo is a commonly used approach to evaluating
expectations of the form $\theta:=\int_{\stany}f(x)\pi(\d x),$ where $\pi$
is an intractable probability measure, e.g. known up to a normalising
constant. One simulates $\lancuch,$ an ergodic Markov chain on $\stany,$
evolving according to a transition kernel $P$ with stationary limiting
distribution $\pi$ and, typically, takes ergodic average as an estimate
of $\theta.$ The approach is justified by asymptotic Markov chain theory,
see e.g. \cite{MeynTw, RobRos}. Metropolis algorithms and Gibbs samplers
(to be described in Section~\ref{sec_setting_counter}) are among the
most common MCMC algorithms, c.f. \cite{RoCa, Liu, RobRos}.

The quality of an estimate produced by an MCMC algorithm depends on
probabilistic properties of the underlying Markov chain. Designing an
appropriate transition kernel $P$ that guarantees rapid convergence
to stationarity and efficient simulation is often a challenging task,
especially in high dimensions. For Metropolis algorithms there are
various optimal scaling results \cite{RobertsGelmanGilks, RobRos_MALA,
Bedard_aap, Bedard_beyond, AtchRobRos_MCMCMC, RobRos_scaling_2001,
RobRos, Rosenthal_proposal} which provide ``prescriptions" of how to do
this, though they typically depend on unknown characteristics of $\pi$.

For random scan Gibbs samplers, a further design decision is choosing the
selection probabilities (i.e., coordinate weightings) which will be used
to select which coordinate to update next.  These are usually chosen to be
uniform, but some recent work \cite{LiuWongKong, Levine05a, Levine05b,
Diaconis_StatSci, sylvia1, sylvia2} has suggested that non-uniform
weightings may sometimes be preferable.


For a very simple toy example to illustrate this issue, suppose
$\stany = [0,1] \times [-100,100]$, with $\pi(x_1,x_2) \propto x_1^{100}
(1+\sin(x_2))$.  Then with respect to $x_1$, this $\pi$ puts almost all of
the mass right up against the line $x_1=1$.  Thus, repeated Gibbs sampler
updates of the coordinate $x_1$ make virtually no difference, and do
not need to be done often at all (unless the functional $f$ of interest
is {\it extremely} sensitive to tiny changes in $x_1$).  By contrast,
with respect to $x_2$, this $\pi$ is a highly multi-modal density with
wide support and many peaks and valleys, requiring many updates to
the coordinate $x_2$ in order to explore the state space appropriately.
Thus, an efficient Gibbs sampler would not update each of $x_1$ and $x_2$
equally often; rather, it would update $x_2$ very often and $x_1$ hardly
at all.  
Of course, in this simple example, it is easy to see directly that $x_1$
should be updated less than $x_2$, and furthermore
such efficiencies would only improve the sampler by approximately a
factor of~2.  However, in a high-dimensional example (c.f.~\cite{sylvia2}),
such issues could be much more significant and also much more difficult
to detect manually.

One promising avenue to address this challenge
is {\it adaptive MCMC algorithms}.
As an MCMC simulation progresses, more and more information about the
target distribution $\pi$ is learned.  Adaptive MCMC attempts to use
this new information to redesign the transition kernel $P$ on the fly,
based on the current simulation output.  That is, the transition kernel
$P_n$ used for obtaining $X_n|X_{n-1}$ may depend on $\{X_0, \dots,
X_{n-1}\}$. 
So, in the above
toy example, a good adaptive Gibbs sampler would somehow automatically
``learn'' to update $x_1$ less often, without requiring the user to
determine this manually (which could be difficult or impossible in a
very high-dimensional problem).

Unfortunately, such adaptive algorithms are only valid if their ergodicity
can be established.  The stochastic process $\lancuch$ for an adaptive
algorithm is no longer a Markov chain; the potential benefit of adaptive
MCMC comes at the price of requiring more sophisticated theoretical
analysis. There is substantial and rapidly growing literature on both
theory and practice of adaptive MCMC (see e.g. \cite{Gilks_adap, Haario,
aro, AndrieuMoulines, Haario2, Kadane, RobRos_JAP, RobRos_ex, LatPhD,
chao2, chao3, chao4, BaiRobRos, Bai2, Bai3, SaksmanVihola, Vihola,
AtchadeFort, AtchadeEtAl}) which includes counterintuitive examples where
$X_n$ fails to converge to the desired distribution $\pi$ (c.f. \cite{aro,
RobRos_JAP, BaiRobRos, LatPhD}), as well as many results guaranteeing
ergodicity under various assumptions.  Most of the previous work on
ergodicity of adaptive MCMC has concentrated on adapting Metropolis and
related algorithms, with less attention paid to ergodicity when adapting
the selection probabilities for random scan Gibbs samplers.


Motivated by such considerations, in the present paper we study the
ergodicity of various types of adaptive Gibbs samplers.  To our knowledge,
proofs of ergodicity for adaptively-weighted Gibbs samplers have
previously been considered only by \cite{Levine_Casella}, and we shall
provide a counter-example below (Example~\ref{ex_stairway_to_heaven_2})
to demonstrate that their main result is not correct.  In view of this,
we are not aware of any valid ergodicity results in the literature that
consider adapting selection probabilities of random scan Gibbs samplers,
and we attempt to fill that gap herein.

This paper is organised as follows.  We begin in
Section~\ref{sec_setting_counter} with basic definitions.
In Section~\ref{sec_counterex} we present a cautionary
Example~\ref{ex_stairway_to_heaven_2}, where a seemingly ergodic
adaptive Gibbs sampler is in fact transient (as we prove formally later
in Section~\ref{sec_counter_proof}) and provides a counter-example
to Theorem~2.1 of \cite{Levine_Casella}. Next, we establish various
positive results for ergodicity of adaptive Gibbs samplers.
In Section~\ref{sec_adap_rand_scan_Gibbs}, we consider adaptive random
scan Gibbs samplers (\texttt{AdapRSG}) which update coordinate selection
probabilities as the simulation progresses; in Section~\ref{sec_adap_MwG},
we consider adaptive random scan Metropolis-within-Gibbs samplers
(\texttt{AdapRSMwG}) which update coordinate selection probabilities
as the simulation progresses; and in Section~\ref{sec_adaptadapt}, we
consider adaptive random scan adaptive Metropolis-within-Gibbs samplers
(\texttt{AdapRSadapMwG}) that update coordinate selection probabilities as
well as proposal distributions for the Metropolis steps -- the case that
corresponds most closely to the adaptations performed in the statistical
genetics work of~\cite{sylvia2}.  In each case, we prove that under
reasonably mild conditions, the adaptive Gibbs samplers are guaranteed
to be ergodic, although our cautionary example does show that it is
important to verify some conditions before applying such algorithms.
Finally, in Section~\ref{sec-componentwise} we consider particular
methods of simultaneously adapting the selection probabilities and
proposal distributions, and prove that in addition to being ergodic,
such algorithms are approximately optimal under certain strong
assumptions.

\section{Preliminaries}
\label{sec_setting_counter}

Gibbs samplers are commonly used MCMC algorithms for sampling from
complicated high-dimensional probability distributions $\pi$ in cases
where the full conditional distributions of $\pi$ are easy to sample from.
To define them,
let $(\stany, \mathcal{B}(\mathcal{X}))$ be an $d-$dimensional state
space where $\stany = \stany_1 \times \dots \times \stany_d$ and write
$X_n \in \stany$ as $X_n = (X_{n,1}, \dots, X_{n,d}).$ We shall use
the shorthand notation
$$
X_{n,-i} \ := \
(X_{n,1}, \dots, X_{n,i-1}, X_{n,i+1},
\dots, X_{n,d})
\, ,
$$
and similarly $\stany_{-i} = \stany_1\times \dots
\times \stany_{i-1}\times \stany_{i+1}\times \dots\times \stany_d$.

Let
$\pi(\cdot| x_{-i})$ denote the conditional distribution of $Z_{i} \, |
\, Z_{-i}=x_{-i}$ where $Z \sim \pi$.
The random scan
Gibbs sampler draws $X_n$ given $X_{n-1}$ (iteratively for
$n=1,2,3,\ldots$) by first choosing one coordinate
at random according to some selection probabilities $\alpha =
(\alpha_1, \dots, \alpha_d)$ (e.g.\ uniformly), and then updating that
coordinate by a draw from its conditional distribution.
More precisely, the Gibbs sampler transition kernel $P =
P_{\alpha}$ is the result of performing the following three steps.

\begin{alg}[\texttt{RSG($\alpha$)}]\label{alg_random_Gibbs}~
\begin{enumerate}
\item Choose coordinate $i \in \{1, \dots, d\}$ according to
selection probabilities $\alpha$, i.e.\ with $\Pr(i=j) = \alpha_j$
\item Draw $Y \sim \pi(\cdot| X_{n-1, -i})$
\item Set $X_{n} : = (X_{n-1,1}, \dots, X_{n-1,i-1}, Y, X_{n-1,i+1}, \dots, X_{n-1,d}).$ 
\end{enumerate}
\end{alg}

Whereas the standard approach is to choose the coordinate $i$ at the
first step uniformly at random, which corresponds to $\alpha = (1/d,
\dots, 1/d)$, this may be a substantial waste of simulation effort if $d$
is large and variability of coordinates differs significantly. This
has been discussed theoretically in \cite{LiuWongKong} and also
observed empirically e.g.\ in Bayesian variable selection for linear
models in statistical genetics \cite{sylvia1, sylvia2}. We
consider a class of adaptive random scan Gibbs samplers where selection
probabilities $\alpha = (\alpha_1, \dots, \alpha_d)$ are subject to
optimization within some subset $\Y \subseteq [0,1]^d$
of possible choices. Therefore a single step of our generic adaptive algorithm
for drawing $X_{n}$ given the trajectory $X_{n-1}, \dots, X_0,$ and
current selection probabilities $\alpha_{n-1} = (\alpha_{n-1,1}, \dots,
\alpha_{n-1,d})$ amounts to the following steps, where $R_n(\cdot)$
is some update rule for $\alpha_n$.

\begin{alg}[\texttt{AdapRSG}]\label{alg_Gibbs_adap}~
\begin{enumerate}
\item Set $\alpha_{n}:=R_n(\alpha_{n-1},X_{n-1}, \dots, X_0) \in \Y$
\item Choose coordinate  $i \in \{1, \dots, d\}$ according to selection probabilities $\alpha_{n}$
\item Draw $Y \sim \pi(\cdot| X_{n-1,-i})$
\item Set $X_{n} : = (X_{n-1,1}, \dots, X_{n-1,i-1}, Y, X_{n-1,i+1}, \dots, X_{n-1,d})$ 
\end{enumerate}
\end{alg}

Algorithm~\ref{alg_Gibbs_adap} defines $P_n$, the transition kernel used at time
$n$, and $\alpha_n$ plays here the role of $\Gamma_n$ in the more general adaptive setting of e.g.\ \cite{RobRos_JAP, BaiRobRos}.
Let $\pi_n = \pi_n(x_0, \alpha_0)$ denote the distribution of $X_n$  induced by Algorithm~\ref{alg_random_Gibbs}~or~\ref{alg_Gibbs_adap}, given starting values $x_0$ and $\alpha_0,$ i.e. for $B \in \mathcal{B}(\mathcal{X}),$ \begin{equation}
\pi_n(B) = \pi_n\big((x_0, \alpha_0),B\big) := \mathbb{P}(X_n \in B | X_0 = x_0, \alpha_0).
\end{equation}
Clearly if one uses Algorithm~\ref{alg_random_Gibbs} then $\alpha_0=\alpha$  remains fixed and $\pi_n(x_0, \alpha)(B) = P_{\alpha}^n(x_0, B)$. By $\|\nu-\mu\|_{TV}$ denote the total variation distance between probability measures $\nu$ and $\mu.$ Let 
\begin{equation}\label{eqn_def_of_T_n} T(x_0, \alpha_0, n):= \|\pi_n(x_0, \alpha_0)-\pi\|_{TV}. \end{equation} 
We call the adaptive Algorithm~\ref{alg_Gibbs_adap} {\it ergodic} if
$T(x_0, \alpha_0, n) \to 0$ for $\pi$-almost every starting state
$x_0$ and all $\alpha_0 \in \Y$.



We shall also consider random scan Metropolis-within-Gibbs
samplers that instead of sampling from the full conditional at step~(2)
of Algorithm~\ref{alg_random_Gibbs} (respectively at step~(3) of
Algorithm~\ref{alg_Gibbs_adap}), perform a single Metropolis
step.  More precisely, given $X_{n-1,-i}$ the $i$-th coordinate
$X_{n-1, i}$ is updated by a draw $Y$ from the proposal distribution
$Q_{X_{n-1,-i}}(X_{n-1, i}, \cdot)$ with the usual
Metropolis acceptance probability for the marginal stationary
distribution $\pi(\cdot| X_{n-1, -i})$.  Such Metropolis-within-Gibbs
algorithms were originally proposed by~\cite{Metropolis} and have been
very widely used.  Versions of this algorithm which adapt the proposal
distributions $Q_{X_{n-1,-i}}(X_{n-1, i}, \cdot)$
were considered by e.g.\ \cite{Haario2, RobRos_ex}, but always with
fixed (usually uniform) coordinate selection probabilities.
If instead the proposal
distributions $Q_{X_{n-1,-i}}(X_{n-1, i}, \cdot)$ remain fixed,
but the selection
probabilities $\alpha_i$ are adapted on the fly, we obtain the following
algorithm
(where $q_{x,-i}(x,y)$ is the density function for
$Q_{x,-i}(x,\cdot)$).

\begin{alg}[\texttt{AdapRSMwG}]~
\begin{enumerate}
\item Set $\alpha_{n}:=R_n(\alpha_{n-1},X_{n-1}, \dots, X_0) \in \Y$
\item Choose coordinate $i \in \{1, \dots, d\}$ according to
selection probabilities $\alpha_n$
\item Draw $Y \sim Q_{X_{n-1,-i}}(X_{n-1, i}, \cdot)$
\item With probability
\begin{equation}
\label{acceptprob}
\min\left(1, \ \ {
\pi(Y| X_{n-1, -i})
\
q_{X_{n-1, -i}}(Y,X_{n-1,i})
\over
\pi(X_{n-1}| X_{n-1, -i})
\
q_{X_{n-1, -i}}(X_{n-1,i},Y)
} \right) \, ,
\end{equation}
accept the proposal and set
$$
X_{n} = (X_{n-1,1}, \dots, X_{n-1,i-1}, Y, X_{n-1,i+1}, \dots, X_{n-1,d})
\, ;
$$
otherwise, reject the proposal and set $X_n = X_{n-1}$.
\end{enumerate}
\end{alg}

\noindent
Ergodicity of
\texttt{AdapRSMwG} is considered in Section~\ref{sec_adap_MwG} below.
Of course, if the proposal distribution
$Q_{X_{n-1,-i}}(X_{n-1, i}, \cdot)$ is symmetric about $X_{n-1}$, then
the $q$ factors in the acceptance probability~(\ref{acceptprob}) cancel
out, and~(\ref{acceptprob}) reduces to the simpler probability
$\min\big(1, \ \pi(Y| X_{n-1, -i}) / \pi(X_{n-1}| X_{n-1, -i})\big)$.

We shall also consider versions of the algorithm in which the proposal
distributions $Q_{X_{n-1,-i}}(X_{n-1, i}, \cdot)$ are also chosen
adaptively, from some family $\{Q_{x_{-i},\gamma}\}_{\gamma
\in \Gamma_i}$ with corresponding density functions $q_{x_{-i},\gamma}$,
as in e.g.\ the statistical genetics
application~\cite{sylvia1, sylvia2}.  Versions of such algorithms with fixed
selection probabilities are considered by e.g.\ \cite{Haario2} and
\cite{RobRos_ex}.  They require additional adaptation parameters
$\gamma_{n,i}$ that are updated on the fly and are allowed to
depend on the past trajectories.  More precisely, if $\gamma_n =
(\gamma_{1,n},\dots, \gamma_{d,n})$ and  $\mathcal{G}_n = \sigma\{X_0,
\dots, X_n, \alpha_0, \dots, \alpha_n, \gamma_0, \dots, \gamma_n\}$,
then the conditional distribution of $\gamma_{n}$ given $\mathcal{G}_{n-1}$
can be specified by the particular algorithm used, via a second update
function $R'_n$.  If we combine such
proposal distribution adaptions with coordinate selection probability
adaptions, this results in a doubly-adaptive algorithm, as follows.

\begin{alg}[\texttt{AdapRSadapMwG}]~
\begin{enumerate}
\item Set $\alpha_{n}:=R_n(\alpha_{n-1},X_{n-1}, \dots, X_0,
\gamma_{n-1},\dots,\gamma_0) \in \Y$
\item Set $\gamma_{n}:=R'_n(\alpha_{n-1},X_{n-1}, \dots, X_0,
\gamma_{n-1},\dots,\gamma_0) \in \Gamma_1 \times \ldots \times \Gamma_n$
\item Choose coordinate $i \in \{1, \dots, d\}$ according to
selection probabilities $\alpha$, i.e.\ with $\Pr(i=j) = \alpha_j$
\item Draw $Y \sim Q_{X_{n-1,-i},\gamma_{n-1}}(X_{n-1, i}, \cdot)$
\item With probability~(\ref{acceptprob}),
$$
\min\left(1, \ \ {
\pi(Y| X_{n-1, -i})
\
q_{X_{n-1, -i}, \gamma_{n-1}}(Y,X_{n-1,i})
\over
\pi(X_{n-1}| X_{n-1, -i})
\
q_{X_{n-1, -i}, \gamma_{n-1}}(X_{n-1,i},Y)
} \right) \, ,
$$
accept the proposal and set
$$
X_{n} = (X_{n-1,1}, \dots, X_{n-1,i-1}, Y, X_{n-1,i+1}, \dots, X_{n-1,d})
\, ;
$$
otherwise, reject the proposal and set $X_n = X_{n-1}$.
\end{enumerate}
\end{alg}

\noindent
Ergodicity of \texttt{AdapRSadapMwG} is considered in
Section~\ref{sec_adaptadapt} below.


\section{A counter-example}
\label{sec_counterex}

Adaptive algorithms destroy the Markovian nature of $\lancuch$, and are
thus notoriously difficult to analyse theoretically.  In particular,
it is easy to be tricked into thinking that a simple adaptive algorithm
``must'' be ergodic when in fact it is not.

For example, Theorem 2.1 of \cite{Levine_Casella} states that ergodicity
of adaptive Gibbs samplers follows from the following two simple
conditions:

\begin{itemize}

\item[(i)] $\alpha_n \to \alpha$ a.s. for some fixed $\alpha \in
(0,1)^d$; and

\item[(ii)] The random scan Gibbs sampler with fixed selection
probabilities $\alpha$ induces an ergodic Markov chain with stationary
distribution $\pi$.

\end{itemize}


Unfortunately, this claim is false, i.e.\ (i) and (ii) alone
do not guarantee ergodicity, as the following example and
proposition demonstrate.  (It seems that in the proof of Theorem~2.1
in~\cite{Levine_Casella}, the same measure is used to represent
trajectories of the adaptive process and of a corresponding non-adaptive
process, which is not correct and thus leads to the error.)



\begin{example} \label{ex_stairway_to_heaven_2} 
Let $\mathbb{N} = \{1,2,\dots\}$, and let the state space
$\stany = \{(i,j) \in \mathbb{N} \times \mathbb{N} : i = j
\textrm{ or } i = j+1\}$, with target distribution given by
$\pi(i,j) \propto j^{-2}.$
On $\stany$, consider a class of adaptive random scan Gibbs samplers for $\pi$,
as defined by Algorithm \ref{alg_Gibbs_adap}, with update rule given by:
\begin{equation}\label{eqn_formula_for_alpha}
R_n\Big(\alpha_{n-1},X_{n-1}=(i,j)\Big) = \left\{\begin{array}{lll} \Big\{\frac{1}{2} + \frac{4}{a_n},\frac{1}{2} - \frac{4}{a_n}\Big\}  & \textrm{if} & i=j,\\ & & \\ \Big\{\frac{1}{2} - \frac{4}{a_n},\frac{1}{2} + \frac{4}{a_n}\Big\}  & \textrm{if} & i=j+1,
\end{array}\right.
\end{equation}
for some choice of the sequence $(a_n)_{n=0}^\infty$ satisfying
$8 <a_n \nearrow \infty$.
\end{example}


Example~\ref{ex_stairway_to_heaven_2} satisfies assumptions~(i)
and~(ii) above.  Indeed, (i) clearly holds since $\alpha_n \to \alpha :=
(\half,\half)$, and (ii) follows immediately from the standard Markov
chain properties of irreducibility and aperiodicity (c.f. \cite{MeynTw,
RobRos}).  However, if $a_n$ increases to $\infty$ slowly enough,
then the example exhibits transient behaviour and is not ergodic.
More precisely, we shall prove the following:



\begin{prop}\label{fact_X_nonergodic} There exists a choice of the
$(a_n)$ for which
the process $\lancuch$ defined in Example~\ref{ex_stairway_to_heaven_2} is
not ergodic. Specifically, 
starting at $X_0=(1,1)$, we have
$\mathbb{P}(X_{n,1} \to \infty) >0$, i.e.\ the
process exhibits transient behaviour with positive probability, so
it does not converge in distribution to any probability measure on
$\stany$.  In particular, $||\pi_n - \pi||_{TV} \nrightarrow 0$.
\end{prop}

\begin{remark}
In fact, we believe that in
Proposition~\ref{fact_X_nonergodic},
$\mathbb{P}(X_{n, 1} \to \infty) = 1$,
though to reduce technicalities we only prove that
$\mathbb{P}(X_{n,1} \to \infty) >0$, which is sufficient
to establish non-ergodicity.
\end{remark}
\medskip

A detailed proof of Proposition~\ref{fact_X_nonergodic} is
presented in Section~\ref{sec_counter_proof}.  We also simulated
Example~\ref{ex_stairway_to_heaven_2} on a computer (with the $(a_n)$ as
defined in Section~\ref{sec_counter_proof}), resulting in the following
trace plot of $X_{n,1}$ which illustrates the transient behaviour since
$X_{n,1}$ increases quickly and steadily as a function of $n$:

\centerline{\includegraphics{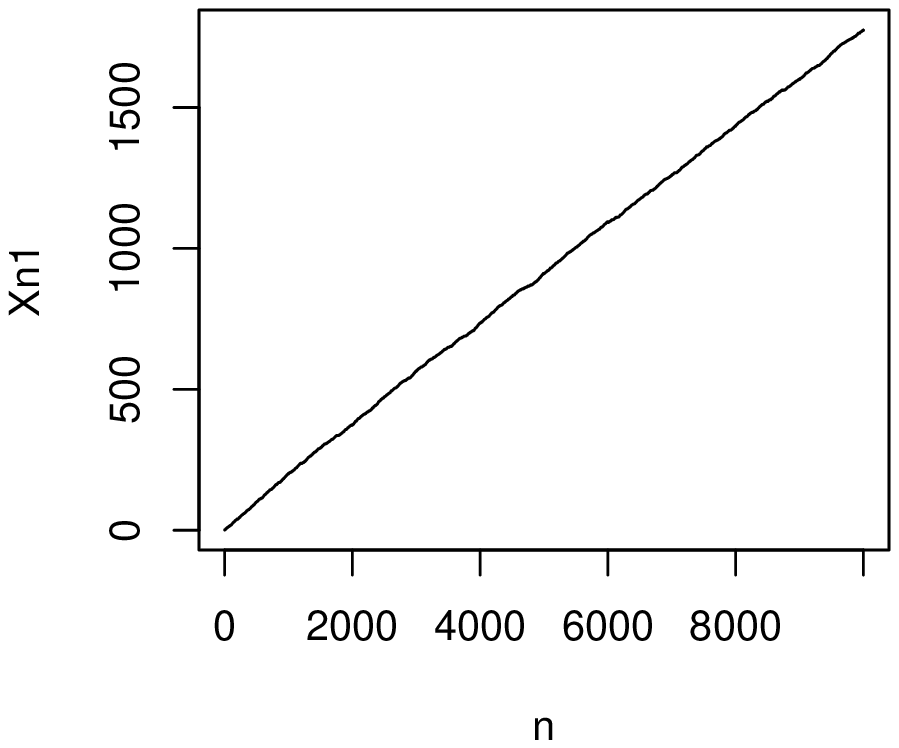}}



\section{Ergodicity of adaptive random scan Gibbs samplers}
\label{sec_adap_rand_scan_Gibbs}


We now present various positive results about ergodicity of adaptive
Gibbs samplers under various assumptions.
Most of our results are specific
to {\it uniformly ergodic} chains.  (Recall that a Markov chain with
transition kernel $P$ is uniformly ergodic if there exist $M<\infty$
and $\rho<1$ s.t. $ \|P^n(x, \cdot) - \pi(\cdot)\|_{TV} \leq M\rho^n \;
\textrm{for every} \; x \in \stany$; see e.g.\ \cite{MeynTw,RobRos} for
this and other notions related to general state space Markov chains.)
In some sense this is a severe restriction, since most MCMC algorithms
arising in statistical applications are not uniformly ergodic.  However,
truncating the variables involved at some (very large) value is usually
sufficient to ensure uniform ergodicity without affecting the statistical
conclusions in any practical sense, so this is not an insurmountable
practical problem.  We do plan to separately consider adaptive Gibbs
samplers in the non-uniformly ergodic case, but that case appears to be
considerably more technical so we do not pursue it further here.

To continue, recall that \texttt{RSG($\alpha$)} stands for random scan Gibbs
sampler with selection probabilities $\alpha$ as defined by
Algorithm~\ref{alg_random_Gibbs}, and \texttt{AdapRSG} is the adaptive
version as defined by Algorithm~\ref{alg_Gibbs_adap}. 
For notation, let $\Delta_{d-1} := \{(p_1, \dots, p_d)  \in
\mathbb{R}^{d}: p_i
\geq 0,\; \sum_{i=1}^d p_i = 1\}$ be the $(d-1)-$dimensional
probability simplex, and let
\begin{equation} \label{Ydef}
\mathcal{Y} := [\varepsilon, 1]^d \cap  \Delta_{d-1}
\end{equation}
for some $0 < \varepsilon \leq 1/d$.
We shall generally assume that all our selection probabilities are in
this set $\Y$, to avoid difficulties arising when one or more of the
selection probabilities approach zero so certain coordinates are
virtually never updated and thus get ``stuck''.

The main result of this section is the following.


\begin{thm}\label{thm_uniform_main}
Let the selection probabilities $\alpha_n \in \Y$ for all $n$,
with $\Y$ as in~(\ref{Ydef}).
Assume that
\begin{itemize}
\item[(a)] $|\alpha_n - \alpha_{n-1}| \to 0$ in probability for fixed starting values $x_0 \in \stany$ and $\alpha_0 \in \mathcal{Y}.$
\item[(b)] there exists $\beta \in \mathcal{Y}$ s.t. \texttt{RSG($\beta$)} is uniformly ergodic.
\end{itemize}
Then \texttt{AdapRSG} is ergodic, i.e.\
\begin{equation}\label{eqn_thm_unif_main_1} T(x_0, \alpha_0, n) \to 0\qquad \textrm{as} \quad  n \to \infty. \end{equation} Moreover, if \begin{itemize} \item[(a')] $\sup_{x_0, \alpha_0}|\alpha_n - \alpha_{n-1}| \to 0 \quad \textrm{in probability,}$\end{itemize} then convergence of \texttt{AdapRSG} is also uniform over all $x_0, \alpha_0,$ i.e.  \begin{equation}\label{eqn_thm_unif_main_2} \sup_{x_0, \alpha_0}T(x_0, \alpha_0, n) \to 0\qquad \textrm{as} \quad  n \to \infty. \end{equation}
\end{thm}

\begin{remark}\label{rem_after_thm_unif}

\begin{enumerate}

\item Assumption \textit{(b)} will typically be
verified for $\beta = (1/d, \dots, 1/d)$; see also
Proposition~\ref{prop_checking_unif_erg_for_G} below.

\item We expect that most adaptive random scan Gibbs samplers will be
designed so that $|\alpha_n - \alpha_{n-1}| \leq a_n$ for every $n\geq
1$, $x_0 \in \stany$, $\alpha_0\in \mathcal{Y}$, and $\omega \in \Omega$,
for some deterministic sequence $a_n \to 0$ (which holds for e.g.\ the
adaptations considered in~\cite{sylvia2}).
In such cases, \textit{(a')} is automatically satisfied.

\item The sequence $\alpha_n$ is not required to
converge, and in particular the amount of adaptation, i.e.\
$\sum_{n=1}^{\infty}|\alpha_n-\alpha_{n-1}|$, is allowed to be infinite.

\item In Example~\ref{ex_stairway_to_heaven_2},
condition $(a')$ is satisfied but condition $(b)$ is not.

\item If we modify Example~\ref{ex_stairway_to_heaven_2} by truncating
the state space to say $\tilde{\stany} = \stany \cap (\{1, \dots,
M\}\times\{1, \dots, M\})$ for some $1<M<\infty,$, then the corresponding
adaptive Gibbs sampler is ergodic, and~(\ref{eqn_thm_unif_main_2}) holds.


\end{enumerate}

\end{remark}

Before we proceed with the proof of Theorem~\ref{thm_uniform_main},
we need some preliminary lemmas, which may be of independent interest.

\begin{lemma}\label{lemma_uniformly_uniformly}
Let $\beta \in \mathcal{Y}$
with $\Y$ as in~(\ref{Ydef}).
If \texttt{RSG($\beta$)} is uniformly ergodic, then also \texttt{RSG($\alpha$)} is uniformly ergodic  for every $\alpha \in \mathcal{Y}$. Moreover there exist $M< \infty$ and $\rho < 1$ s.t. $\;\sup_{x_0 \in \stany, \alpha \in \mathcal{Y}}T(x_0, \alpha, n) \leq M\rho^n \to 0.$
\end{lemma}

\begin{proof} Let $P_{\beta}$ be the transition kernel of \texttt{RSG($\beta$)}. It is well known that for uniformly ergodic Markov chains the whole state space $\stany$ is small (c.f. Theorem 5.2.1 and 5.2.4 in \cite{MeynTw} with their $\psi=\pi$). Thus there exists $s>0,$ a probability measure $\mu$ on $(\stany, \borel)$ and a positive integer $m,$ s.t. for every $x \in \stany,$ 
\begin{equation} \label{eqn_minorization}
P_{\beta}^m(x, \cdot) \geq s\mu(\cdot).
\end{equation}
Fix $\alpha \in \mathcal{Y}$ and let $$r:= \min_i \frac{\alpha_i}{\beta_i}.$$ Since $\beta \in \mathcal{Y},$ we have $1 \geq r \geq \frac{\varepsilon}{1-(d-1)\varepsilon} > 0$  and $P_{\alpha}$ can be written as a mixture of transition kernels of two random scan Gibbs samplers, namely $$ P_{\alpha} =  rP_{\beta} + (1-r)P_{q}, \qquad \textrm{where} \quad q= \frac{\alpha- r\beta}{1-r}.$$
This combined with (\ref{eqn_minorization}) implies 
\begin{eqnarray} \nonumber
 P_{\alpha}^m(x, \cdot) &\geq & r^mP_{\beta}^m(x, \cdot) \;  \geq \; r^ms\mu(\cdot) \\ \label{eqn_minorization_mixture_for_ranodm_Gibbs}  &\geq & \Big(\frac{\varepsilon}{1-(d-1)\varepsilon}\Big)^m s \mu(\cdot) \qquad \textrm{for every} \quad x \in \stany.  \qquad
\end{eqnarray}
By Theorem 8 of \cite{RobRos} condition (\ref{eqn_minorization_mixture_for_ranodm_Gibbs}) implies \begin{equation}\label{eqn_gibbs_unif}  \|P_{\alpha}^n(x, \cdot) - \pi(\cdot)\|_{TV} \leq \bigg(1- \Big(\frac{\varepsilon}{1-(d-1)\varepsilon}\Big)^m s \bigg)^{\lfloor n/m\rfloor} \qquad \textrm{for all} \quad x \in \stany.\end{equation}
Since the right hand side of (\ref{eqn_gibbs_unif}) does not depend on $\alpha,$ the claim follows.
\end{proof}

\begin{lemma}\label{lemma_uniformly_Lipshitz} Let $P_{\alpha}$ and $P_{\alpha'}$ be random scan Gibbs samplers using selection probabilities $\alpha, \alpha' \in \mathcal{Y}:=[\varepsilon, 1-(d-1)\varepsilon]^d$
for some $\varepsilon>0$. Then
\begin{equation}\label{eqn_lemma_dimin_adap_by_weights}
\|P_{\alpha}(x, \cdot) - P_{\alpha'}(x, \cdot)\|_{TV}  \leq \frac{|\alpha- \alpha'|}{\varepsilon + |\alpha- \alpha'|} \leq \frac{|\alpha- \alpha'|}{\varepsilon}.
\end{equation} \end{lemma} 

\begin{proof}
Let $\delta:= |\alpha - \alpha'|.$ Then $r:=\min_i\frac{\alpha'_i}{\alpha_i} \geq \frac{\varepsilon}{\varepsilon+\max_i|\alpha_i - \alpha'_i|} \geq \frac{\varepsilon}{\varepsilon+\delta}$ and reasoning as in the proof of Lemma~\ref{lemma_uniformly_uniformly} we can write $P_{\alpha'} = rP_{\alpha} + (1-r)P_q$ for some $q$ and compute
\begin{eqnarray}\nonumber \|P_{\alpha}(x, \cdot) - P_{\alpha'}(x, \cdot)\|_{TV} & = & \|(rP_{\alpha} + (1-r)P_{\alpha}) - (rP_{\alpha} + (1-r)P_q)\|_{TV} \\ \nonumber & = & (1-r) \|P_{\alpha} - P_q\|_{TV} \leq \frac{\delta}{\varepsilon+\delta},
 \end{eqnarray} as claimed. \end{proof}

\begin{cor}
$P_{\alpha}(x, B)$ as a function of $\alpha$ on $\mathcal{Y}$ is Lipshitz with Lipshitz constant $1/\varepsilon$ for every fixed set $B \in \borel.$ 
\end{cor}

\begin{cor} \label{cor_diminish_adap}
If $|\alpha_n-\alpha_{n-1}| \to 0$ in probability, then also $\sup_{x \in \stany}\|P_{\alpha_n}(x, \cdot) - P_{\alpha_{n-1}}(x, \cdot)\|_{TV} \to 0$ in probability.
\end{cor}

\begin{proof}[Proof of Theorem~\ref{thm_uniform_main}] We conclude the result from Theorem~1 of \cite{RobRos_JAP} that requires simultaneous uniform ergodicity and diminishing adaptation. Simultaneous uniform ergodicity results from combining assumption (b) and Lemma~\ref{lemma_uniformly_uniformly}. Diminishing adaptation results from assumption (a) with Corollary \ref{cor_diminish_adap}. Moreover note that Lemma~\ref{lemma_uniformly_uniformly} is uniform in $x_0$ and $\alpha_0$ and $(a')$ yields uniformly diminishing adaptation again by Corollary~\ref{cor_diminish_adap}. A look into the proof of Theorem~1 \cite{RobRos_JAP} reveals that this suffices for the uniform part of Theorem~\ref{thm_uniform_main}.\end{proof}

Finally, we note that verifying uniform ergodicity of a random scan Gibbs
sampler, as required by assumption $(b)$ of Theorem~\ref{thm_uniform_main},
may not be straightforward. Such issues have been investigated in e.g.\
\cite{RobertsPolson} and more recently in relation to the parametrization
of hierarchical models (see \cite{PapaGareth_parametr} and references
therein).  In the following proposition, we show that to verify uniform
ergodicity of any random scan Gibbs sampler, it suffices to verify
uniform ergodicity of the corresponding systematic scan Gibbs sampler
(which updates the coordinates $1, 2, \ldots, d$ in sequence rather than
select coordinates randomly).


\begin{prop}\label{prop_checking_unif_erg_for_G}
Let $\alpha \in \mathcal{Y}$ with $\Y$ as in~(\ref{Ydef}).
If the systematic scan Gibbs sampler is uniformly ergodic, then so
is \texttt{RSG($\alpha$)}.
\end{prop}

\begin{proof} Let $$P = P_1P_2\cdots P_d$$ be the transition kernel of the uniformly ergodic systematic scan Gibbs sampler, where $P_i$ stands for the step that updates coordinate $i.$ By the minorisation condition characterisation, there exist $s>0,$ a probability measure $\mu$ on $(\stany, \borel)$ and a positive integer $m,$ s.t. for every $x \in \stany,$ 
\begin{equation} \nonumber
P^m(x, \cdot) \geq s\mu(\cdot).
\end{equation}
However, the probability that the random scan Gibbs sampler $P_{1/d}$ in its $md$ subsequent steps will update the coordinates in exactly the same order is $(1/d)^{md}>0.$ Therefore the following minorisation condition holds for the random scan Gibbs sampler. 
\begin{equation} \nonumber
P_{1/d}^{md}(x, \cdot) \geq (1/d)^{md}s\mu(\cdot).
\end{equation}
We conclude that \texttt{RSG($1/d$)} is uniformly ergodic,
and then by Lemma~\ref{lemma_uniformly_uniformly} it
follows that \texttt{RSG($\alpha$)} is uniformly ergodic for
any $\alpha\in\Y$.
\end{proof}

\section{Adaptive random scan Metropolis-within-Gibbs}
\label{sec_adap_MwG}


In this section we consider random scan Metropolis-within-Gibbs
sampler algorithms.  Thus, given $X_{n-1,-i}$, the $i$-th coordinate
$X_{n-1, i}$ is updated by a draw $Y$ from the proposal distribution
$Q_{X_{n-1,-i}}(X_{n-1, i}, \cdot)$ with the usual Metropolis acceptance
probability for the marginal stationary distribution $\pi(\cdot|
X_{n-1, -i})$.  Here, we consider Algorithm~\texttt{AdapRSMwG}, where
the proposal distributions $Q_{X_{n-1,-i}}(X_{n-1, i}, \cdot)$ remain
fixed, but the selection probabilities $\alpha_i$ are adapted on the fly.
We shall prove ergodicity of such algorithms under some circumstances.
(The more general algorithm \texttt{AdapRSadapMwG} is then considered
in the following section.)

To continue, let $P_{x_{-i}}$ denote the resulting Metropolis transition
kernel for obtaining $X_{n, i}|X_{n-1, i}$ given $X_{n-1, -i} = x_{-i}$.
We shall require the following assumption.


\begin{ass}\label{assu_Metrop_steps_uniformly}
For every $i \in \{1, \dots, d\}$ the transition kernel $P_{x_{-i}}$ is uniformly ergodic for every $x_{-i} \in \stany_{-i}.$ Moreover there exist $s_i >0$ and an integer $m_i$ s.t. for every  $x_{-i} \in \stany_{-i}$ there exists a probability measure $\nu_{x_{-i}}$ on $(\stany_i,  \mathcal{B}(\stany_i)),$ s.t. $$P_{x_{-i}}^{m_i}(x_i, \cdot) \geq s_i \nu_{x_{-i}}(\cdot) \qquad \textrm{for every} \quad x_i \in \stany_i.$$
\end{ass}

We have the following counterpart of Theorem \ref{thm_uniform_main}.

\begin{thm}\label{thm_unif_MwG}
Let $\alpha_n \in \mathcal{Y}$ for all $n$,
with $\Y$ as in~(\ref{Ydef}).
Assume that
\begin{itemize}
\item[(a)] $|\alpha_n - \alpha_{n-1}| \to 0$ in probability for fixed starting values $x_0 \in \stany$ and $\alpha_0 \in \mathcal{Y}.$
\item[(b)] there exists $\beta \in \mathcal{Y}$ s.t. \texttt{RSG($\beta$)} is uniformly ergodic.
\item[(c)] Assumption~\ref{assu_Metrop_steps_uniformly} holds.
\end{itemize}
Then \texttt{AdapRSMwG} is ergodic, i.e.\
\begin{equation}\label{eqn_thm_unif_MwG_1} T(x_0, \alpha_0, n) \to 0
\qquad \textrm{as} \quad  n \to \infty.
\end{equation}
Moreover, if \begin{itemize} \item[(a')] $\sup_{x_0, \alpha_0}|\alpha_n - \alpha_{n-1}| \to 0 \quad \textrm{in probability,}$\end{itemize} then convergence of \texttt{AdapRSMwG} is also uniform over all $x_0, \alpha_0,$ i.e.\
\begin{equation}
\label{eqn_thm_unif_MwG_2} \sup_{x_0, \alpha_0}T(x_0, \alpha_0, n) \to 0\qquad \textrm{as} \quad  n \to \infty.
\end{equation}
\end{thm}



\begin{remark}\label{rem_after_thm_unif_MwG}~
Remarks~\ref{rem_after_thm_unif}.1--\ref{rem_after_thm_unif}.3 still apply.
Also,
assumption~\ref{assu_Metrop_steps_uniformly} can easily be verified in some cases of interest, e.g. \begin{enumerate} 
\item  Independence samplers are essentially uniformly ergodic if and only if the candidate density is bounded below by a multiple of the stationary density, i.e. $q(\d x) \geq s \pi(\d x)$ for some $s>0,$ c.f. \cite{MengersenTweedie}. 
\item The Metropolis-Hastings algorithm with continuous and positive proposal density $q(\cdot, \cdot)$ and bounded target density $\pi$ is uniformly ergodic if the state space is compact, c.f. \cite{MeynTw, RobRos}. 
\end{enumerate}
\end{remark}

To prove Theorem~\ref{thm_unif_MwG} we build on the approach of \cite{RobRos_hybrid_AAP}. In particular recall the following notion of strong uniform ergodicity.
\begin{defi}
We say that a transition kernel $P$ on $\stany$ with stationary distribution $\pi$ is $(m, s)-$\emph{strongly uniformly ergodic}, if for some $s> 0$ and positive integer $m$ $$P^m(x, \cdot) \geq s \pi(\cdot) \qquad \textrm{for every} \quad x \in \stany.$$ 
Moreover, we will say that a family of Markov chains $\big\{P_{\gamma}\big\}_{\gamma \in \Gamma}$ on $\stany$ with stationary distribution $\pi$ is $(m,s)-$\emph{simultaneously strongly uniformly ergodic}, if for some $s> 0$ and positive integer $m$ $$P_{\gamma}^m(x, \cdot) \geq s \pi(\cdot) \qquad \textrm{for every} \quad x \in \stany \quad \textrm{and} \quad \gamma \in \Gamma.$$  
\end{defi}
By Proposition~1 in \cite{RobRos_hybrid_AAP}, if a Markov chain is both uniformly ergodic and reversible, then it is strongly uniformly ergodic. The following lemma improves over this result by controlling both involved parameters.
\begin{lemma}\label{lem_simultaneous_strong_unif_erg}
Let $\mu$ be a probability measure on $\stany$, let $m$ be a positive integer and let $s>0.$ If a reversible transition kernel $P$ satisfies the condition $$P^m(x, \cdot) \geq  s\mu(\cdot) \qquad \textrm{for every} \quad x \in \stany,$$ then it is $\left(\left(\left\lfloor \frac{\log(s/4)}{\log(1-s)}\right\rfloor +2\right)m, \frac{s^2}{8}\right)-$strongly uniformly ergodic.
\end{lemma}

\begin{proof}
By Theorem 8 of \cite{RobRos} for every $A \in \borel$ we have \begin{equation}\nonumber
\|P^n(x, A) - \pi(A)\|_{TV} \leq (1-s)^{\lfloor n/m \rfloor},
\end{equation} And in particular \begin{equation}\label{eqn_lem_proof_TVdiff_small}
\|P^{km}(x, A) - \pi(A)\|_{TV} \leq s/4 \qquad \textrm{for} \quad k \geq \frac{\log(s/4)}{\log(1-s)}.
\end{equation} Since $\pi$ is stationary for $P,$ we have $\pi(\cdot) \geq s\mu(\cdot)$ and thus an upper bound for the Radon-Nikodym derivative \begin{equation} \label{eqn_RN_bound_mu_pi} \d \mu/\d \pi \leq 1/s.\end{equation} Moreover by reversibility $$ \pi(\d x) P^m(x, \d y) = \pi(\d y) P^m(y, \d x) \geq \pi(\d y)s \mu(\d x) $$ and consequently \begin{equation} \label{eqn_lower_for_P_m} P^m(x, \d y) \geq s\big(\mu (\d x)/\pi(\d x)\big) \pi(\d y).\end{equation}  Now define $$A := \{x \in \stany: \mu(\d x)/\pi(\d x) \geq 1/2\}$$ Clearly $\mu(A^c) \leq 1/2.$ Therefore by (\ref{eqn_RN_bound_mu_pi}) we have $$1/2 \leq \mu(A) \leq  (1/s)\pi(A)$$ and hence $\pi(A) \geq s/2.$ Moreover (\ref{eqn_lem_proof_TVdiff_small}) yields $$P^{km}(x, A) \geq s/4 \qquad \textrm{for} \quad k: = \left\lfloor \frac{\log(s/4)}{\log(1-s)}\right\rfloor +1.$$
And with $k$ defined above by (\ref{eqn_lower_for_P_m}) we have
\begin{eqnarray} \nonumber P^{km+m}(x, \cdot) &=& \int_{\stany} P^{km}(x, \d z)P^m(z, \cdot) \geq \int_{A} P^{km}(x, \d z)P^m(z, \cdot) \\ & \geq & \int_{A} P^{km}(x, \d z)(s/2) \pi(\cdot) \geq (s^2/8) \pi(\cdot).\nonumber\end{eqnarray}
This completes the proof. \end{proof}

We will need the following generalization of Lemma~\ref{lemma_uniformly_uniformly}.
\begin{lemma}\label{lem_RSMwG_uniform}
Let $\beta \in \mathcal{Y}$
with $\Y$ as in~(\ref{Ydef}).
If \texttt{RSG($\beta$)} is uniformly ergodic then there exist $s'>0$ and a positive integer $m'$ s.t. the family $\big\{\texttt{RSG($\alpha$)}\big\}_{\alpha \in \mathcal{Y}}$ is $(m', s')-$simultaneously strongly uniformly ergodic.
\end{lemma}
\begin{proof} $P_{\beta}(x, \cdot)$ is uniformly ergodic and reversible, therefore by Proposition~1 in \cite{RobRos_hybrid_AAP} it is $(m,s_1)-$strongly uniformly ergodic for some $m$ and $s_1.$ Therefore, and arguing as in the proof of Lemma~\ref{lemma_uniformly_uniformly}, c.f. (\ref{eqn_minorization_mixture_for_ranodm_Gibbs}), there exist $s_2 \geq \big(\frac{\varepsilon}{1-(d-1)\varepsilon}\big)^m, $ s.t. for every $\alpha \in \mathcal{Y}$ and every $x \in \stany$
\begin{equation}
P_{\alpha}^m(x, \cdot) \geq s_2 P_{\beta}^m(x, \cdot)  \geq s_1 s_2 \pi(\cdot).
\end{equation}
Set $m'=m$ and $s' = s_1s_2.$
\end{proof}

\begin{proof}[Proof of Theorem~\ref{thm_unif_MwG}] We proceed as in the proof of Theorem~\ref{thm_uniform_main}, i.e. establish diminishing adaptation and simultaneous uniform ergodicity and conclude (\ref{eqn_thm_unif_MwG_1}) and (\ref{eqn_thm_unif_MwG_2}) from Theorem 1 of \cite{RobRos_JAP}. Observe that Lemma~\ref{lemma_uniformly_Lipshitz} applies for random scan Metropolis-within-Gibbs algorithms exactly the same way as for random scan Gibbs samplers. Thus diminishing adaptation results from assumption (a) and Corollary~\ref{cor_diminish_adap}. To establish simultaneous uniform ergodicity, observe that by Assumption~\ref{assu_Metrop_steps_uniformly} and Lemma~\ref{lem_simultaneous_strong_unif_erg} the Metropolis transition kernel for $i$th coordinate i.e. $P_{x_{-i}}$ has stationary distribution $\pi(\cdot|x_{-i})$ and is $\left(\left(\left\lfloor \frac{\log(s_i/4)}{\log(1-s_i)}\right\rfloor +2\right)m_i, \frac{s_i^2}{8}\right)-$strongly uniformly ergodic. Moreover by Lemma~\ref{lem_RSMwG_uniform} the family \texttt{RSG($\alpha$)}, $\alpha \in \mathcal{Y}$ is $(m', s')-$strongly uniformly ergodic, therefore by Theorem~2 of \cite{RobRos_hybrid_AAP} the family of random scan Metropolis-within-Gibbs samplers with selection probabilities $\alpha \in \mathcal{Y},$ \texttt{RSMwG($\alpha$)}, is $(m_*, s_*)-$simultaneously strongly uniformly ergodic with $m_*$ and $s_*$ given as in \cite{RobRos_hybrid_AAP}.   
\end{proof}

We close this section with the following alternative version of Theorem~\ref{thm_unif_MwG}.


\begin{thm}\label{thm_unif_MwG_alternative}
Let $\alpha_n \in \mathcal{Y}$ for all $n$,
with $\Y$ as in~(\ref{Ydef}).
Assume that \begin{itemize}
\item[(a)] $|\alpha_n - \alpha_{n-1}| \to 0$ in probability for fixed starting values $x_0 \in \stany$ and $\alpha_0 \in \mathcal{Y}.$
\item[(b)] there exists $\beta \in \mathcal{Y}$ s.t. \texttt{RSMwG($\beta$)} is uniformly ergodic.
\end{itemize}
Then \texttt{AdapRSMwG} is ergodic, i.e. \begin{equation}\label{eqn_thm_unif_MwG_alternative_1} T(x_0, \alpha_0, n) \to 0\qquad \textrm{as} \quad  n \to \infty. \end{equation} Moreover, if \begin{itemize} \item[(a')] $\sup_{x_0, \alpha_0}|\alpha_n - \alpha_{n-1}| \to 0 \quad \textrm{in probability,}$\end{itemize} then convergence of \texttt{AdapRSMwG} is also uniform over all $x_0, \alpha_0,$ i.e.  \begin{equation}\label{eqn_thm_unif_MwG_alternative_2} \sup_{x_0, \alpha_0}T(x_0, \alpha_0, n) \to 0\qquad \textrm{as} \quad  n \to \infty. \end{equation}
\end{thm}

\begin{proof} Diminishing adaptation results from assumption (a) and Corollary~\ref{cor_diminish_adap}. Simultaneous uniform ergodicity can be established as in the proof of Lemma~\ref{lemma_uniformly_uniformly}. The claim follows from Theorem~1 of \cite{RobRos_JAP}.
\end{proof}

\begin{remark}
Whereas the statement of Theorem~\ref{thm_unif_MwG_alternative} may be useful in specific examples, typically condition (b), the uniform ergodicity of a random scan Metropolis-within-Gibbs sampler, will be not available and establishing it will involve conditions required by Theorem~\ref{thm_unif_MwG}. 
\end{remark}

\section{Adaptive random scan adaptive Metropolis-within-Gibbs}
\label{sec_adaptadapt}

In this section, we consider the adaptive random scan adaptive
Metropolis-within-Gibbs algorithm \texttt{AdapRSadapMwG}, that
updates both selection probabilities of the Gibbs kernel and proposal
distributions of the Metropolis step.  Thus, given $X_{n-1,-i}$, the
$i$-th coordinate $X_{n-1, i}$ is updated by a draw $Y$ from a proposal
distribution $Q_{X_{n-1,-i}, \;\gamma_{n,i}}(X_{n-1, i}, \cdot)$ with
the usual acceptance probability.
This doubly-adaptive algorithm has been used by e.g.~\cite{sylvia2}
for an application in statistical genetics.  As with adaptive
Metropolis algorithms, the adaption of the proposal distributions in
this setting is motivated by optimal scaling results for random walk
Metropolis algorithms \cite{RobertsGelmanGilks, RobRos_MALA, Bedard_aap,
Bedard_beyond, AtchRobRos_MCMCMC, RobRos_scaling_2001, RobRos, RobRos_ex,
Rosenthal_proposal}.

Let $P_{x_{-i},\; \gamma_{n,i}}$ denote the resulting
Metropolis transition kernel for obtaining $X_{n, i}|X_{n-1,
i}$ given $X_{n-1, -i} = x_{-i}.$ We will prove ergodicity
of this generalised algorithm using tools from the previous
section. Assumption~\ref{assu_Metrop_steps_uniformly} must be reformulated
accordingly, as follows.

\begin{ass}\label{assu_adap_Metrop_steps_uniformly}
For every $i \in \{1, \dots, d\},$ $x_{-i} \in \stany_{-i}$ and $\gamma_i \in \Gamma_i,$ the transition kernel $P_{x_{-i},\; \gamma_i}$  is uniformly ergodic. Moreover there exist $s_i >0$ and an integer $m_i$ s.t. for every  $x_{-i} \in \stany_{-i}$ and $\gamma_i \in \Gamma_i$ there exists a probability measure $\nu_{x_{-i}, \; \gamma_i}$ on $(\stany_i,  \mathcal{B}(\stany_i)),$ s.t. $$P_{x_{-i},\;\gamma_i}^{m_i}(x_i, \cdot) \geq s_i \nu_{x_{-i}, \; \gamma_i}(\cdot) \qquad \textrm{for every} \quad x_i \in \stany_i.$$
\end{ass}

 We have the following counterpart of Theorems \ref{thm_uniform_main} and \ref{thm_unif_MwG}.


\begin{thm}\label{thm_unif_aMwG}
Let $\alpha_n \in \mathcal{Y}$ for all $n$,
with $\Y$ as in~(\ref{Ydef}).
Assume that
\begin{itemize}
\item[(a)] $|\alpha_n - \alpha_{n-1}| \to 0$ in probability for fixed starting values $x_0 \in \stany$ and $\alpha_0 \in \mathcal{Y}.$
\item[(b)] there exists $\beta \in \mathcal{Y}$ s.t. \texttt{RSG($\beta$)} is uniformly ergodic.
\item[(c)] Assumption~\ref{assu_adap_Metrop_steps_uniformly} holds.
\item[(d)] The Metropolis-within-Gibbs kernels exhibit diminishing adaptation, i.e. for every $i \in \{1, \dots, d\}$ the $\mathcal{G}_{n+1}$ measurable random variable $$\sup_{x\in \mathcal{X}}\|P_{x_{-i},\; \gamma_{n+1,i}}(x_i, \cdot) - P_{x_{-i},\; \gamma_{n,i}}(x_i, \cdot)\|_{TV} \to 0 \textrm{ in probability, as } n \to \infty,$$ for fixed starting values $x_0 \in \stany$ and $\alpha_0 \in \mathcal{Y}.$ 
\end{itemize}
Then \texttt{AdapRSadapMwG} is ergodic, i.e. \begin{equation}\label{eqn_thm_unif_aMwG_1} T(x_0, \alpha_0, n) \to 0\qquad \textrm{as} \quad  n \to \infty. \end{equation} Moreover, if \begin{itemize} \item[(a')] $\sup_{x_0, \alpha_0}|\alpha_n - \alpha_{n-1}| \to 0 \quad \textrm{in probability,}$
\item[(d')] $\sup_{x_0, \alpha_0} \sup_{x\in \mathcal{X}}\|P_{x_{-i},\; \gamma_{n+1,i}}(x_i, \cdot) - P_{x_{-i},\; \gamma_{n,i}}(x_i, \cdot)\|_{TV} \to 0 \textrm{ in probability,}$
\end{itemize} then convergence of \texttt{AdapRSadapMwG} is also uniform over all $x_0, \alpha_0,$ i.e.  \begin{equation}\label{eqn_thm_unif_aMwG_2} \sup_{x_0, \alpha_0}T(x_0, \alpha_0, n) \to 0\qquad \textrm{as} \quad  n \to \infty. \end{equation}
\end{thm}

\begin{remark}
Remarks~\ref{rem_after_thm_unif}.1--\ref{rem_after_thm_unif}.3 still
apply.  And, Remark~\ref{rem_after_thm_unif_MwG} applies for verifying
Assumption~\ref{assu_adap_Metrop_steps_uniformly}. Verifying condition
$(d)$ is discussed after the proof.
\end{remark}

\begin{proof} We again proceed by establishing diminishing adaptation and
simultaneous uniform ergodicity and concluding the result from Theorem
1 of \cite{RobRos_JAP}. To establish simultaneous uniform ergodicity
we proceed as in the proof of Theorem~ \ref{thm_unif_MwG}. Observe
that by Assumption~\ref{assu_adap_Metrop_steps_uniformly} and
Lemma~\ref{lem_simultaneous_strong_unif_erg} every adaptive
Metropolis transition kernel for $i$th coordinate i.e. $P_{x_{-i},\;
\gamma_i}$ has stationary distribution $\pi(\cdot|x_{-i})$ and is
$\left(\left(\left\lfloor \frac{\log(s_i/4)}{\log(1-s_i)}\right\rfloor
+2\right)m_i, \frac{s_i^2}{8}\right)-$strongly uniformly
ergodic. Moreover, by Lemma~\ref{lem_RSMwG_uniform} the family
\texttt{RSG($\alpha$)}, $\alpha \in \mathcal{Y}$ is $(m', s')-$strongly
uniformly ergodic, therefore by Theorem~2 of \cite{RobRos_hybrid_AAP}
the family of random scan Metropolis-within-Gibbs samplers with selection
probabilities $\alpha \in \mathcal{Y}$ and proposals indexed by $\gamma
\in \Gamma,$ is $(m_*, s_*)-$simultaneously strongly uniformly ergodic
with $m_*$ and $s_*$ given as in \cite{RobRos_hybrid_AAP}.

For diminishing adaptation we write
\begin{eqnarray}
&& \sup_{x\in \mathcal{X}}\|P_{\alpha_n,\;\gamma_n}(x, \cdot) - P_{\alpha_{n-1},\;\gamma_{n-1}}(x, \cdot) \|_{TV}  \;\; \leq \nonumber \\ 
&& \qquad \qquad \qquad \qquad \qquad \qquad \sup_{x\in \mathcal{X}}\|P_{\alpha_n,\;\gamma_n}(x, \cdot) - P_{\alpha_{n-1},\;\gamma_{n}}(x, \cdot) \|_{TV}  \nonumber \\ \nonumber && \qquad \qquad \qquad \qquad \qquad \qquad  
+ \; \sup_{x\in \mathcal{X}}\|P_{\alpha_{n-1},\;\gamma_n}(x, \cdot) - P_{\alpha_{n-1},\;\gamma_{n-1}}(x, \cdot) \|_{TV} 
\end{eqnarray}
The first term above converges to $0$ in probability by Corollary~\ref{cor_diminish_adap} and assumption (a). The second term
\begin{eqnarray}
&& \sup_{x\in \mathcal{X}}\|P_{\alpha_{n-1},\;\gamma_n}(x, \cdot) - P_{\alpha_{n-1},\;\gamma_{n-1}}(x, \cdot) \|_{TV} \;\; \leq \nonumber \\
&&\qquad \qquad \qquad \qquad  \sum_{i=1}^d \alpha_{n-1, i}\sup_{x\in \mathcal{X}}\|P_{x_{-i},\; \gamma_{n+1,i}}(x_i, \cdot) - P_{x_{-i},\; \gamma_{n,i}}(x_i, \cdot)\|_{TV} \nonumber 
\end{eqnarray}
converges to $0$ in probability as a mixture of terms that converge to $0$ in probability. \end{proof}

The following lemma can be used to verify assumption $(d)$ of
Theorem~\ref{thm_unif_aMwG}; see also Example~\ref{example_diminish}
below.

\begin{lemma} \label{lemma_diminish_Q_P}
Assume that the adaptive proposals exhibit diminishing adaptation i.e. for every $i \in \{1, \dots, d\}$ the $\mathcal{G}_{n+1}$ measurable random variable $$\sup_{x\in \mathcal{X}}\|Q_{x_{-i},\; \gamma_{n+1,i}}(x_i, \cdot) - Q_{x_{-i},\; \gamma_{n,i}}(x_i, \cdot)\|_{TV} \to 0 \textrm{ in probability, as } n \to \infty,$$ for fixed starting values $x_0 \in \stany$ and $\alpha_0 \in \mathcal{Y}.$

Then any of the following conditions
\begin{itemize}
\item[(i)] The Metropolis proposals have symmetric densities, i.e. $$q_{x_{-i},\; \gamma_{n,i}}(x_i, y_i)\; = \;  q_{x_{-i},\; \gamma_{n,i}}(y_i, x_i),$$
\item[(ii)] $\stany_i$ is compact for every $i$, $\pi$ is continuous, everywhere positive and bounded,
\end{itemize}
implies condition $(d)$ of Theorem~\ref{thm_unif_aMwG}.
\end{lemma}
\begin{proof}
Let $P_1,$ $P_2$ denote transition kernels and $Q_1,$ $Q_2$ proposal kernels of two generic Metropolis algorithms for sampling from $\pi$ on arbitrary state space $\stany.$ To see that $(i)$ implies $(d)$ we check that \begin{equation}
\|P_1(x, \cdot)-P_2(x, \cdot)\|_{TV} \; \leq \; 2\|Q_1(x, \cdot)-Q_2(x, \cdot)\|_{TV} .\nonumber
\end{equation}
Indeed, the acceptance probability $$\alpha(x,y)=\min\Big\{1, \frac{\pi(y)}{\pi(x)}\Big\} \; \in [0,1]$$ does not depend on the proposal, and for any $x\in \stany$ and $A \in \borel$ we compute   \begin{eqnarray}
|P_1(x, A)-P_2(x, A)| & \leq & \left|\int_A \alpha(x,y)\big(q_1(y)-q_2(y)\big) \d y\right| \nonumber \\
&& \qquad + \;\mathbb{I}_{\{x \in A\}} \left|\int_{\mathcal{X}} \big(1-\alpha(x,y)\big)\big(q_1(y)-q_2(y)\big) \d y\right| \nonumber \\
& \leq & 2\|Q_1(x, \cdot)-Q_2(x, \cdot)\|_{TV} . \nonumber
\end{eqnarray}
Condition $(ii)$ implies that there exists $K < \infty,$ s.t. $\pi(y)/\pi(x) \leq K$ for every $x,y \in \stany.$ To conclude that $(d)$ results from $(ii)$ note that \begin{equation}|\min\{a,b\} -\min\{c,d\}| \; < \; |a-c| + |b-d|\label{eqn_abs_val}\end{equation}
and recall acceptance probabilities $\alpha_i(x,y)=\min\Big\{1, \frac{\pi(y)q_i(y,x)}{\pi(x)q_i(x,y)}\Big\}.$ Indeed for any $x\in \stany$ and $A \in \borel$ using (\ref{eqn_abs_val}) we have \begin{eqnarray}
|P_1(x, A)-P_2(x, A)| & \leq & \bigg| \int_A \bigg( \min\Big\{q_1(x,y), \frac{\pi(y)}{\pi(x)} q_1(y,x)\Big\} \nonumber \\ & & \qquad \qquad \qquad - \min\Big\{q_2(x,y), \frac{\pi(y)}{\pi(x)} q_2(y,x) \Big\}\bigg) \d y \bigg| \nonumber \\ &&
+\; \mathbb{I}_{\{x \in A\}} \bigg|\int_{\mathcal{X}} \Big(\big(1-\alpha_1(x,y)\big)q_1(x,y) \nonumber \\ && \qquad \qquad \qquad - \big(1-\alpha_2(x,y)\big)q_2(x,y)\Big) \d y\bigg| \nonumber \\
& \leq & 4(K+1) \|Q_1(x, \cdot)-Q_2(x, \cdot)\|_{TV} \nonumber
\end{eqnarray}
And the claim follows since a random scan Metropolis-within-Gibbs sampler is a mixture of Metropolis samplers.
\end{proof}

We now provide an example to show that diminishing adaptation of proposals
as in Lemma~\ref{lemma_diminish_Q_P} does not necessarily imply condition
$(d)$ of Theorem~\ref{thm_unif_aMwG}, so some additional assumption is
required, e.g.\ (i) or (ii) of Lemma~\ref{lemma_diminish_Q_P}.

\begin{example} \label{example_diminish}
Consider a sequence of Metropolis algorithms with transition kernels
$P_1, P_2, \dots$ designed for sampling from $\pi(k) = p^k(1-p)$ on
$\stany  = \{0,1,\dots\}.$ The transition kernel $P_n$ results from
using proposal kernel $Q_n$ and the standard acceptance rule, where
\begin{eqnarray} \nonumber
Q_n(j,k) & = & q_n(k) \;\; := \;\; \left\{ \begin{array}{lll} p^k
\big(\frac{1}{1-p}-p^n+p^{2n}\big)^{-1} & \textrm{ for } & k \neq n,
\\ p^{2n} \big(\frac{1}{1-p}-p^n+p^{2n}\big)^{-1}  & \textrm{ for } &
k = n.
\end{array} \right.
\end{eqnarray}
Clearly \begin{eqnarray} \sup_{j \in \stany} \| Q_{n+1}(j,\cdot) -
Q_{n}(j,\cdot)\|_{TV} & = & q_{n+1}(n) - q_n(n) \to 0.
\nonumber \end{eqnarray}
However
\begin{eqnarray} \nonumber
\sup_{j \in \stany} \| P_{n+1}(j,\cdot) - P_{n}(j,\cdot)\|_{TV} & \geq  &
P_{n+1}(n,0) - P_{n}(n,0) \\  &=& \nonumber  \min\Big\{q_{n+1}(0),
\frac{\pi(0)}{\pi(n)}q_{n+1}(n)\Big\} \\ \nonumber  && \qquad \qquad
- \;\;\min\Big\{q_{n}(0),  \frac{\pi(0)}{\pi(n)}q_{n}(n)\Big\} \\
\nonumber & = & q_{n+1}(0) - q_n(0)p^n \to 1-p \neq 0.
\end{eqnarray}
\end{example}


\section{A specific Metropolis-within-Gibbs adaptive choice}
\label{sec-componentwise}

As an application of the previous section, we discuss a particular
method of adapting the $\alpha_i$ selection probabilities for the
doubly-adaptive Metropolis-within-Gibbs algorithms.  We are motivated by
two closely-related
componentwise
adaptation algorithms, from~\cite{Haario2} and from Section~3 of
\cite{RobRos_ex}. Briefly, these algorithms
use a deterministic scan Metropolis-within-Gibbs sampler and perform a
random walk Metropolis step for updating coordinate $i$ by
proposing a normal increment to $X_{n-1,i},$ i.e. the proposal
$Y_{n,i} \sim N(X_{n-1,i}, \sigma^2_{n,i}).$ The proposal variance
$\sigma^2_{n,i}$ is subject to adaptation. Haario et al. in
\cite{Haario2} use \begin{eqnarray}\label{eqn_prop_var_HST}
\sigma^{2,\textrm{HST}}_{n,i} & = &(2.4)^2(s^2_{n,i} +
0.05),\end{eqnarray} where $s^2_{n,i}$ is the sample variance of
$X_{0,i}, \dots, X_{n-1,i},$ whereas Roberts and Rosenthal in
\cite{RobRos_ex} take
\begin{eqnarray}\label{eqn_prop_var_RR}\sigma^{2,\textrm{RR}}_{n,i} &
= & e^{ls_i},\end{eqnarray} and $ls_i$ is updated every batch of 50
iterations by adding or subtracting $\delta(n) = O(n^{-1/2}).$
Specifically, $ls_i$ is increased by $\delta(n)$ if the fraction of
acceptances of variable $i$ was more then $0.44$ on the last batch and
decreased if it was less.

Both rules have theoretical motivation, c.f.
\cite{RobRos_scaling_2001}; $\sigma^{2,\textrm{HST}}_{n,i}$ comes from
diffusion limit considerations in infinite dimensions and
$\sigma^{2,\textrm{RR}}_{n,i}$ is motivated by one dimensional
Gaussian target densities. Conclusions drawn in this very special
situations are observed empirically to be robust in a wide range of
examples that are neither high-dimensional nor Gaussian
\cite{RobRos_scaling_2001, RobRos_ex}.

In this section, we use a random scan Gibbs sampler instead of a
deterministic scan, and optimise the coordinate selection probabilities
$\alpha_i$ simultaneously with proposal variances. We aim at minimizing
the asymptotic variance.  Under certain strong conditions
(Assumption~\ref{assu_for_componentwise}) that allow for illustrative
analysis and explicit calculations, we shall provide approximately optimal
adaptions for the $\alpha_i$ in equations (\ref{eqn_adap_alpha_def_HST})
and (\ref{eqn_adap_alpha_def_RR}) below, and shall prove ergodicity of
the corresponding algorithms in Theorem~\ref{componentwise_thm}. More
general adaptation algorithms for random scan Gibbs samplers have
been investigated by others (e.g.\ \cite{LiuWongKong, Levine_Casella,
Levine05a, Levine05b}).

\begin{ass}\label{assu_for_componentwise}
The following conditions hold. \begin{itemize}
\item[(i)] The stationary distribution on $\mathcal{X} = \mathbb{R}^d$
is of the product form \begin{eqnarray}\label{assu_eqn_product_target}
\pi(x) & = & \prod_{i=1}^d C_i \, g(C_i x_i),
\end{eqnarray} where $g$ is a one dimensional density and $C_i,$ $i=1,
\dots, d,$ are unknown, strictly positive constants.
\item[(ii)] The second moment of $g$ exists, i.e.\ $\sigma^2:=
Var_{g}Z < \infty$.
\item[(iii)] The one-dimensional
random walk Metropolis algorithm with $N(x, 1)$
proposal distributions and target density $g$ is uniformly ergodic.
\end{itemize}
\end{ass}

We consider an adaptive random scan adaptive random walk
Metropolis-within-Gibbs algorithm $\texttt{AdapRSadapMwG},$ with Gaussian
proposals, for estimating expectation of a linear target function
\begin{eqnarray}\label{eqn_example_section_f_linear}
f(x) & = & a_0 + \sum_{i=1}^d a_i x_i.
\end{eqnarray}

A random scan Gibbs sampler for a target density of product form
(\ref{assu_eqn_product_target}) is uniformly ergodic, therefore
arguing as in the proof of Theorem~\ref{thm_unif_aMwG}, under
Assumption~\ref{assu_for_componentwise} a random scan
Metropolis-within-Gibbs with $N(x, 1)$ proposals is uniformly ergodic.
Moreover, by $(ii),$ function $f$ defined in
(\ref{eqn_example_section_f_linear}) is square integrable and the
Markov chain CLT holds, i.e. for any initial distribution of $X_0$
\begin{eqnarray}\label{eqn_CLT}
n^{-1/2} \bigg(\sum_{k=0}^{n-1} f(X_i) - n\mathbb{E}_{\pi} f(X)\bigg)
& \to & N(0,\sigma^2_{\textrm{as}}), \qquad \textrm{ as} \quad n \to
\infty,
\end{eqnarray}
where the asymptotic variance $\sigma^2_{\textrm{as}} < \infty$ can be
written as
\begin{eqnarray}\label{eqn_as_var_1}
\sigma^2_{\textrm{as}} & = & \tau_f Var_{\pi}f(X), \qquad \textrm{
and} \\ \label{eqn_int_autocor_1} \tau_f & = & 1 +
2\sum_{k=1}^{\infty}Cor_{\pi}(f(X_0), f(X_k)),
\end{eqnarray}
is the stationary integrated autocorrelation time. Markov chain CLTs
and asymptotic variance formulae are discussed e.g. in \cite{RobRos,
HagRos, BeLaLa}. Note that under
Assumption~\ref{assu_for_componentwise} the asymptotic variance
decomposes and some explicit computations are possible.
\begin{eqnarray}\label{eqn_as_var2}
\sigma^2_{\textrm{as}} & = & \sum_{i=1}^d \sigma^2_{\textrm{as}, i}
\;\; = \;\; \sum_{i=1}^d \tau_{f,i} Var_{\pi, i}f(X), \qquad \textrm{
where}\qquad  \\ \label{eqn_int_autocor_2} \tau_{f,i} & = & 1 +
2\sum_{k=1}^{\infty}Cor_{\pi}(X_{0,i}, X_{k,i}), \qquad \textrm{ and}
\\ \label{eqn_as_var_components_2} Var_{\pi, i}f(X) & := &
Var_{\pi}(a_i X_{0,i}) \;\; = \;\;   \frac{a_i^2}{C_i^2}\sigma^2.
\end{eqnarray}
To compute $\tau_{f,i}$ for a random scan Metropolis-within-Gibbs
sampler in the present setting, we focus solely on coordinate $i,$
i.e. the Markov chain $X_{n,i},$ $n = 0,1, \dots$ Due to the product
form of $\pi,$ the distribution of $X_{n,i}|X_{n}$ does not depend on
$X_{n, -i}.$ Let $P_i$ be the transition kernel that describes the
dynamics of $X_{n, i},$ $n = 0,1, \dots$ and let $\alpha =
(\alpha_{1}, \dots, \alpha_{d})$ denote the (fixed) selection
probabilities.  We write $P_i$ as a mixture \begin{eqnarray}
\label{eqn_mixture_repres_for_coordinatewise_RSMwG}
P_i(x_i, \cdot) & = & (1-\alpha_{i})\textrm{Id} +
\alpha_{i}P^{\textrm{Metrop}}_{i}(x_i, \cdot),
\end{eqnarray}
where $\textrm{Id}$ denotes the identity kernel and
$P^{\textrm{Metrop}}_{i}$ performs a single Metropolis step for the
target distribution $C_i g(C_i x).$ Thus $P_i$ is a lazy version of
$P^{\textrm{Metrop}}_{i},$ since it performs a
$P^{\textrm{Metrop}}_{i}$ step if coordinate $i$ is selected with
probability $\alpha_{i}$ and an identity step otherwise. We will use
Lemma~\ref{lemma_as_var_of_lazy} below, which is a general result
about asymptotic variance of lazy reversible Markov chains. Suppose $$
h \in L_0^2(\pi):= \{h \in L^2(\pi): \pi h = 0\},$$ and denote
$$\sigma^2_{h,H} := \lim_{n \to \infty} \frac{1}{n} Var
\bigg(\sum_{i=0}^{n-1} h(Z_i)\bigg),$$ where $Z_0, Z_2, \dots$ is a
Markov chain with  transition kernel $H$ and initial distribution
$\pi$ that is stationary for $H.$

\begin{lemma}\label{lemma_as_var_of_lazy} Let $P$ be a reversible
transition kernel with stationary measure $\pi.$ Let $\delta \in
(0,1)$ and by $P_{\delta}$ denote its lazy version $$P_{\delta} \; =
\; (1-\delta)\textrm{Id} + \delta P.$$ Then \begin{eqnarray}
\label{eqn_as_var_of_lazy}
\sigma^2_{h, P_{\delta}} & = & \frac{1}{\delta}\sigma^2_{h,P} +
\frac{1-\delta}{\delta} \pi h^2.
\end{eqnarray}
\end{lemma}

\begin{proof} The proof is based on the functional analytic approach
(see e.g. \cite{KipnisVaradhan, RobRos_electr}). A reversible
transition kernel $P$ with invariant distribution $\pi$ is a
self-adjoint operator on $L_0^2(\pi)$ with spectral radius bounded by
1. By the spectral decomposition theorem for self adjoint operators,
for each $h \in L^2_0(\pi)$ there exists a finite positive measure
$E_{h,P}$ on $[-1,1],$ such that \begin{eqnarray} \nonumber
\langle h, P^n h \rangle & = & \int_{[-1,1]} x^n E_{h,P}(dx),
\end{eqnarray}
for all integers $n \geq 0.$ Thus in particular \begin{eqnarray}
\label{eqn_spectral_in_proof_stat_var}
\pi h^2 & = & \int_{[-1,1]}1 E_{h,P}(dx), \\
\sigma^2_{h,P} & = & \int_{[-1,1]}\frac{1+x}{1-x}E_{h,P}(dx).
\label{eqn_spectral_in_proof_as_var}
\end{eqnarray}
Since $$P_{\delta}^n \;\;= \;\;\big((1-\delta)\textrm{Id} + \delta
P\big)^n = \sum_{k=0}^n \binom{n}{k}(1-\delta)^k\delta^{n-k}P^{n-k},$$
we have
\begin{eqnarray}\nonumber
\langle h, P^n_{\delta}h \rangle & = & \int_{[-1,1]} ((1-\delta) +
\delta x)^n E_{h,P}(dx), \qquad \textrm{ and consequently} \\
\nonumber
\sigma^2_{h,P_{\delta}} & = & \int_{[-1,1]}\frac{1+1-\delta + \delta
x}{1-1 + \delta - \delta x}E_{h,P}(dx) \\ \nonumber
& = & \int_{[-1,1]}\frac{1}{\delta}\bigg( \frac{1+ x}{1- x} + 1-
\delta \bigg)E_{h,P}(dx)  \\ \nonumber & = & \frac{1}{\delta}
\int_{[-1,1]}\frac{1+x}{1-x}E_{h,P}(dx)     +
\frac{1-\delta}{\delta}\int_{[-1,1]}1 E_{h,P}(dx),
\end{eqnarray}
as claimed. \end{proof}

Let \begin{eqnarray}\label{eqn_as_var_of_i_Metrop_kernel}
\tilde{\sigma}^2_{\textrm{as},i} & = & \tilde{\tau}_{f,i}
Var_{\pi}(a_i X_{0,i}) \;\; = \;\; \tilde{\tau}_{f,i}
\frac{a_i^2}{C_i^2} \sigma^2,
\end{eqnarray} be the asymptotic variance of the Metropolis kernel
$P_i^{\textrm{Metrop}}$ defined in
(\ref{eqn_mixture_repres_for_coordinatewise_RSMwG}). Here
$\tilde{\tau}_{f,i}$ is its stationary integrated autocorrelation
time.  From Lemma~\ref{lemma_as_var_of_lazy} we have the following
formula for $\sigma^2_{as, i}$ of (\ref{eqn_as_var2}).
\begin{eqnarray} \nonumber
\sigma^2_{\textrm{as}, i} & = & \frac{1}{\alpha_i}
\tilde{\sigma}^2_{\textrm{as}, i} + \frac{1- \alpha_i}{\alpha_i}
\frac{a_i^2}{C_i^2} \sigma^2, \qquad \textrm{ hence} \\
\label{eqn_autocor_relations} \tau_{f,i} & = & \frac{1}{\alpha_i}
\tilde{\tau}_{\textrm{as}, i} + \frac{1- \alpha_i}{\alpha_i}.
\end{eqnarray}
Now we take advantage of the fact that $f$ is linear and of the actual
adaptation of the proposal variances performed by both versions,
i.e.\ HST and RR. Namely, they aim at minimizing their integrated
autocorrelation time $\tilde{\tau}_{\textrm{as}, i}.$ Under
Assumption~\ref{assu_for_componentwise} the conditional distributions
are equally shaped up to the scaling constant $C_i.$ However the
adaptive algorithm will learn $C_i$ and adjust the proposal variance
accordingly. We conclude that after an initial learning period the
following proportionality relation will hold approximately $$
\sigma^{2,\textrm{HST}}_{n,i} \;\; \propto \;\;
\sigma^{2,\textrm{RR}}_{n,i} \;\;\propto \;\;  1/C_i^2,$$ and also the
stationary integrated autocorrelation times for the adapted
$P_i^{\textrm{Metrop}}$ will be close to the (unknown) optimal value,
say $T,$ i.e. \begin{eqnarray}\label{eqn_tilde_tau_const}
\tilde{\tau}_{\textrm{as}, i} & \approx & T.\end{eqnarray} Typically
$T \gg 1,$ hence we can approximately write (using
(\ref{eqn_autocor_relations}), (\ref{eqn_as_var2}),
(\ref{eqn_int_autocor_2}), (\ref{eqn_as_var_components_2}) and
(\ref{eqn_mixture_repres_for_coordinatewise_RSMwG}))
\begin{eqnarray} \nonumber
\tau_{f,i} & \approx & T/\alpha_i, \\ \label{eqn_final_variance_coord}
\sigma^2_{\textrm{as}, i}  & \approx & \frac{T a_i^2}{C_i^2 \alpha_i}\sigma^2,
\qquad \textrm{ and finally} \\ \label{eqn_final_variance}
\sigma^2_{\textrm{as}} & \approx & T \sigma^2 \sum_{i=1}^d
\frac{a_i^2}{C_i^2 \alpha_i} \;\; \propto \;\; \sum_{i=1}^d
\frac{\sigma^{2,\textrm{HST}}_{n,i}  a_i^2}{\alpha_i} \;\; \propto
\;\; \sum_{i=1}^d \frac{\sigma^{2,\textrm{RR}}_{n,i}
a_i^2}{\alpha_i}.
\end{eqnarray}
The last expression is minimised for
\begin{eqnarray}\label{eqn_final_choice_of_alpha_i}
\alpha_i & \propto & \Big( \sigma^{2,\textrm{HST}}_{n,i}  a_i^2
\Big)^{1/2} \;\; \propto \;\; \Big( \sigma^{2,\textrm{RR}}_{n,i}
a_i^2 \Big)^{1/2},
\end{eqnarray}
which yields a very intuitive prescription for adapting selection
probabilities, namely by setting \begin{eqnarray}
\label{eqn_adap_alpha_def_HST}
\alpha_{n,i}^{\textrm{HST}} & := & \frac{\Big(
\sigma^{2,\textrm{HST}}_{n,i}  a_i^2 \Big)^{1/2}}{\sum_{k=1}^d \Big(
\sigma^{2,\textrm{HST}}_{n,k}  a_k^2 \Big)^{1/2}} \qquad \textrm{for
the HST version of \cite{Haario2}, and}\qquad \quad  \\
\label{eqn_adap_alpha_def_RR}
\alpha_{n,i}^{\textrm{RR}} & := & \frac{\Big(
\sigma^{2,\textrm{RR}}_{n,i}  a_i^2 \Big)^{1/2}}{\sum_{k=1}^d \Big(
\sigma^{2,\textrm{RR}}_{n,k}  a_k^2 \Big)^{1/2}} \qquad \textrm{ for
the RR version of \cite{RobRos_ex}.}
\end{eqnarray}
The above argument shows:
{\em (\ref{eqn_adap_alpha_def_HST}) and
(\ref{eqn_adap_alpha_def_RR}) are approximately optimal choices of
adaptive selection probabilities for these algorithms, at least for
target densities of the form (\ref{assu_eqn_product_target}).}

We next prove ergodicity of these algorithms.
Let HST-algorithm denote an \texttt{AdapRSadapMwG} that uses
(\ref{eqn_prop_var_HST}) for updating proposal variances and
(\ref{eqn_adap_alpha_def_HST}) for updating selection probabilities.
Similarly let RR-algorithm follow (\ref{eqn_prop_var_RR}) and
(\ref{eqn_adap_alpha_def_RR}) with additional restriction for $ls_i$
to stay in $[-M, M]$ for some fixed, large $M<\infty$ (which
technically plays the role of $0.05$ in  (\ref{eqn_prop_var_HST}) for
the HST-algorithm).

\begin{thm} \label{componentwise_thm}
Under Assumption~\ref{assu_for_componentwise} the HST- and
RR-algorithms are ergodic.
\end{thm}
\begin{proof} It is enough to check that the assumptions of
Theorem~\ref{thm_unif_aMwG} are satisfied. We do this for the
HST-algorithm; the proof for the RR-algorithm follows in the same way.
Condition (b) is immediately implied by
Assumption~\ref{assu_for_componentwise}~(i), since~(b) requires only that
the full Gibbs sampler is uniformly ergodic, which is obvious
for a product target density of the form~(\ref{assu_eqn_product_target}).
Next, observe that
Assumption~\ref{assu_for_componentwise} (iii) implies that the support
of $C_i g(C_i x),$ say $S_{g,i},$ is bounded, therefore the sample
variance estimate in (\ref{eqn_prop_var_HST}) is bounded from
above and for the HST-algorithm, for every $i\in \{1, \dots, d\},$
\begin{eqnarray}\label{eqn_sigma_bounded}
\sigma_{n,i}^{2,\textrm{HST}} & \in & [(2.4)^2 0.05, \; K_i] \; =: \;
S_{\sigma, i}\end{eqnarray} for some $K_i < \infty.$ Thus (a) holds
since the denominator in (\ref{eqn_adap_alpha_def_HST}) is bounded
from below and the change in sample variance \begin{eqnarray}
\label{eqn_diff_in_sigma} \sigma_{n,i}^{2,\textrm{HST}}
-\sigma_{n+1,i}^{2,\textrm{HST}} & = & O(n^{-1}).\end{eqnarray}
Condition (d) results from (\ref{eqn_diff_in_sigma}),
(\ref{eqn_sigma_bounded}) and Lemma~\ref{lemma_diminish_Q_P}~(i). We
are left with (c). Let $\phi_{\sigma}(\cdot)$ denote the density function
of $N(0,\sigma^2).$ Since $$\sup_{i\in \{1, \dots, d\}; x,y \in
S_{g,i};\sigma_1, \sigma_2 \in S_{\sigma, i}}
\phi_{\sigma_1}(x-y)/\phi_{\sigma_2}(x-y) < \infty,$$ the
Radon-Nikodym derivative of all pairs of proposals for every
coordinate is bounded and hence
Assumption~\ref{assu_adap_Metrop_steps_uniformly} is implied again by
Assumption~\ref{assu_for_componentwise}~(iii).
\end{proof}

%

\begin{remark}
\begin{enumerate}
\item Condition $(i)$ of Assumption~\ref{assu_for_componentwise} is
very restrictive, however it already proved extremely helpful in
understanding high dimensional MCMC algorithms via diffusion limits
\cite{RobertsGelmanGilks, RobRos_scaling_2001, Bedard_aap},
and conclusions drawn under $(i)$ are empirically observed to be
robust even if the condition is violated. It is essential to
investigate its robustness also in the Gibbs sampler setting.
\item Minor generalisations to $(i)$ are straightforward, e.g. our
conclusions hold for $\mathcal{X} = \prod_{i=1}^d \mathcal{X}_i,$
where $\mathcal{X}_i = \mathbb{R}^k.$
\item Condition $(iii)$ of Assumption~\ref{assu_for_componentwise} is
required to ensure asymptotic validity of our algorithm by
Theorem~\ref{thm_unif_aMwG}. We will report separately on ergodicity
of adaptive random scan Gibbs samplers in the non-uniform case.
\end{enumerate}
\end{remark}



\section{Proof of Proposition~\ref{fact_X_nonergodic}}
\label{sec_counter_proof}

The analysis of Example~\ref{ex_stairway_to_heaven_2} is somewhat
delicate since the process is both time and space inhomogeneous (as
are most nontrivial adaptive MCMC algorithms). To establish
Proposition~\ref{fact_X_nonergodic}, we will define a couple of auxiliary
stochastic process. Consider the following one dimensional process
$\lancucht$ obtained from $\lancuch$ by $$\tilde{X}_n := X_{n,1} +
X_{n,2}-2.$$ Clearly $\tilde{X}_n - \tilde{X}_{n-1} \in \{-1,0,1\},$
moreover $X_{n,1}\to \infty$ and $X_{n,2} \to \infty$ if and only if
$\tilde{X}_n \to \infty.$ Note that the dynamics of $\lancucht$ are
also both time and space inhomogeneous.

We will also use an auxiliary random-walk-like space homogeneous process $$S_0 = 0 \quad \textrm{and} \quad S_n: = \sum_{i=1}^n Y_i,\quad \textrm{for } n\geq 1,$$ where $Y_1, Y_2, \dots$ are independent random variables taking values in $\{-1,0,1\}.$ Let the distribution of $Y_n$ on $\{-1,0,1\}$ be \begin{equation}\label{eqn_distr_of_Y} \nu_n \;:= \;\left\{\frac{1}{4}-\frac{1}{a_n} \;,\;\frac{1}{2}\;,\;\frac{1}{4}+\frac{1}{a_n} \right\}.\end{equation}

We shall couple $\lancucht$ with $\lancuchs,$ i.e. define them on the same probability space $\{\Omega, \mathcal{F}, \mathbb{P}\},$ by specifying the joint distribution of $(\tilde{X}_n, S_n)_{n \geq 0}$ so that the marginal distributions remain unchanged. We describe the details of the construction later. Now define \begin{equation}\label{eqn_def_Omega_XS} \Omega_{\tilde{X}\geq S}:= \{ \omega \in \Omega: \tilde{X}_n(\omega) \geq S_n(\omega) \;\textrm{ for every }\;n\}\end{equation} and \begin{equation}\label{eqn_def_Omega_infty} \Omega_{\infty}:= \{ \omega \in \Omega: S_n(\omega) \to \infty\}.\end{equation} Clearly, if $\omega \in \Omega_{\tilde{X}\geq S}\cap\Omega_{\infty},$ then $\tilde{X}_n(\omega) \to \infty.$ In the sequel we show that for our coupling construction \begin{equation}\label{eqn_goal} \mathbb{P}(\Omega_{\tilde{X}\geq S}\cap\Omega_{\infty}) > 0.\end{equation}

We shall use the Hoeffding's inequality for $S_k^{k+n} := S_{k+n} - S_k.$ Since $Y_n \in [-1,1],$ it yields for every $t> 0,$ \begin{equation}
\label{eqn_Hoeff_first_time} 
\mathbb{P}(S_k^{k+n} - \mathbb{E}S_k^{k+n} \leq -nt) \leq \exp \{-\frac{1}{2}nt^2\}.
\end{equation} 
Note that $\mathbb{E}Y_n = 2/a_n$ and thus $\mathbb{E}S_{k}^{k+n} = 2\sum_{i=k+1}^{k+n}1/a_i.$ The following choice for the sequence $a_n$ will facilitate further calculations. Let
\begin{eqnarray*}
b_0 & = & 0, \\
b_1 & = & 1000, \\
b_{n}& = &b_{n-1} \Big(1 + \frac{1}{10 + \log(n)}\Big), \qquad \textrm{for} \quad n \geq 2\\
c_n & = & \sum_{i=0}^{n}b_n,\\
a_n & = & 10+\log(k),\, \quad \qquad \qquad \qquad \textrm{for} \quad c_{k-1} < n \leq c_{k}.
\end{eqnarray*}
\begin{remark} To keep notation reasonable we ignore the fact that $b_n$ will not be an integer. It should be clear that this does not affect proofs, as the constants we have defined, i.e. $b_1$ and $a_1$ are bigger then required.\end{remark}
\begin{lemma}\label{lemma_S_n_above} Let $Y_n$ and $S_n$ be as defined above and let \begin{eqnarray}
\Omega_1 & := & \Big\{\omega \in \Omega: S_k =k \quad \textrm{for every} \quad 0 < k \leq c_1\Big\}.\\ \label{eqn_lemma_S_n_above} 
\Omega_n &:=& \Big\{\omega \in \Omega: S_k \geq \frac{b_{n-1}}{2} \quad \textrm{for every} \quad c_{n-1} < k \leq c_n\Big\} \;\quad \textrm{for } \; n\geq 2.\qquad 
\end{eqnarray}
Then \begin{equation} \label{eqn_lemma_Y_n_above}
\mathbb{P}\bigg( \bigcap_{n=1}^{\infty} \Omega_n\bigg) > 0.
\end{equation}
\end{lemma}
\begin{remark}
Note that $b_n \nearrow \infty$ and therefore $\bigcap_{n=1}^{\infty} \Omega_n \subset \Omega_{\infty}.$
\end{remark}
\begin{proof}
With positive probability, say $p_{1,S},$ we have $Y_1 = \dots = Y_{1000}=1$ which gives $S_{c_1} = 1000 = b_1.$ Hence $\mathbb{P}(\Omega_1) = p_{1,S} >0.$ Moreover recall that $S_{c_{n-1}}^{c_{n}}$ is a sum of $b_n$ i.i.d. random variables with $\mathbb{E}S_{c_{n-1}}^{c_{n}} = \frac{2b_n}{10+\log(n)}.$ Therefore for every $n \geq 1$ by Hoeffding's inequality with $t= 1/(10+\log(n)),$ we can also write
\begin{equation}
\nonumber 
\mathbb{P}\Big(S_{c_{n-1}}^{c_n} \leq \frac{b_n}{10+\log(n)}\Big) \leq \exp \bigg\{-\frac{1}{2}\frac{b_n}{(10+\log(n))^2}\bigg\} = :p_n.
\end{equation} 
Therefore using the above bound iteratively we obtain \begin{equation} \label{eqn_S_proof_1}
\mathbb{P}(S_{c_1} = b_1,\; S_{c_n} \geq b_n \;\; \textrm{for every} \;\; n \geq 2) \;\geq \;p_{1,S}\prod_{n=2}^{\infty} (1-p_n).
\end{equation}
Now consider the minimum of $S_k$ for $c_{n-1} < k \leq c_n$ and $n \geq 2.$ The worst case is when the process $S_k$ goes monotonically down and then monotonically up for $c_{n-1} < k \leq c_n.$ By the choice of $b_n,$ equation (\ref{eqn_S_proof_1}) implies also 
\begin{equation}
\mathbb{P}\bigg( \bigcap_{n=1}^{\infty} \Omega_n\bigg) \geq p_{1,S} \prod_{n=2}^{\infty} (1-p_n).
\end{equation}
Clearly in this case \begin{equation}\label{eqn_sum_p_n_infty} p_{1,S}\prod_{n=2}^{\infty}(1-p_n) > 0 \quad \Leftrightarrow \quad \sum_{n=1}^{\infty} \log(1-p_n) > -\infty \quad \Leftrightarrow \quad \sum_{n=1}^{\infty} p_n < \infty.\quad
\end{equation} 
We conclude (\ref{eqn_sum_p_n_infty}) by comparing $p_n$ with $1/n^2.$ We show that there exists $n_0$ such that for $n \geq n_0$ the series $p_n$ decreases quicker then the series $1/n^2$ and therefore $p_n$ is summable. We check that \begin{equation}\label{eqn_comparison}\log \frac{p_{n-1}}{p_{n}} > \log \frac{n^2}{(n-1)^2} \qquad \textrm{for} \quad n \geq n_0.\end{equation}
Indeed
\begin{eqnarray*}
\log \frac{p_{n-1}}{p_{n}} & = & -\frac{1}{2}\left(\frac{b_{n-1}}{(10+\log(n-1))^2} - \frac{b_n}{(10+\log(n))^2}\right)\\
 & = & \frac{b_{n-1}}{2}\left(\frac{11 + \log(n)}{(10+\log(n))^3} - \frac{1}{(10+\log(n-1))^2} \right)\\
& = &  \frac{b_{n-1}}{2}\left( \frac{(11 + \log(n))(10+\log(n-1))^2 - (10+\log(n))^3}{(10+\log(n))^3(10+\log(n-1))^2}\right).
\end{eqnarray*}
Now recall that $b_{n-1}$ is an increasing sequence. Moreover the enumerator can be rewritten as $$(10 + \log(n))\Big((10+\log(n-1))^2 - (10+\log(n))^2\Big)+ (10 + \log(n-1))^2,$$
now use $a^2-b^2 = (a+b)(a-b)$ to identify the leading term $(10 + \log(n-1))^2.$ Consequently there exists a constant $C$ and $n_0 \in \mathbb{N}$ s.t. for $n \geq n_0$
\begin{eqnarray*}
\log \frac{p_{n-1}}{p_{n}} & \geq & \frac{C}{(10+\log(n))^3} \; > \; \frac{2}{n-1} \; > \; \log \frac{n^2}{(n-1)^2}.
\end{eqnarray*} 
Hence $ \sum_{n=1}^{\infty} p_n < \infty$ follows.
\end{proof}

Now we will describe the coupling construction of $\lancucht$ and $\lancuchs$. We already remarked that $\bigcap_{n=1}^{\infty} \Omega_n \subset \Omega_{\infty}.$ We will define a coupling that implies also \begin{equation}\label{eqn_goal_1}\mathbb{P}\Bigg(\bigg(\bigcap_{n=1}^{\infty} \Omega_n\bigg) \cap \Omega_{\tilde{X}\geq S}\Bigg) \geq C \mathbb{P}\bigg(\bigcap_{n=1}^{\infty} \Omega_n\bigg) \qquad \textrm{for some universal} \quad C>0,\;\end{equation} and therefore \begin{equation}\label{eqn_goal_2}\mathbb{P}\left(\Omega_{\tilde{X}\geq S}\cap \Omega_{\infty}\right) > 0.\end{equation}
Thus nonergodicity of $\lancuch$ will follow from Lemma \ref{lemma_S_n_above}. We start with the following observation.
\begin{lemma}\label{lemma_coupling_above} There exists a coupling of $\tilde{X}_{n} - \tilde{X}_{n-1}$ and $Y_n,$ such that
\begin{enumerate}
\item[(a)] For every $n \geq 1$ and every value of $\tilde{X}_{n-1}$ \begin{equation}
\mathbb{P}(\tilde{X}_{n} - \tilde{X}_{n-1} = 1, Y_n = 1) \geq \mathbb{P}(\tilde{X}_{n} - \tilde{X}_{n-1} = 1)\mathbb{P}(Y_n = 1),\qquad
\end{equation}
\item[(b)] Write even or odd $\tilde{X}_{n-1}$ as $\tilde{X}_{n-1} = 2i-2$ or $\tilde{X}_{n-1} = 2i-3$ respectively.  If  $2i - 8 \geq a_n$ then the following implications hold a.s. \begin{eqnarray}\label{eqn_coupling_above_1} Y_n = 1 & \; \Rightarrow \; & \tilde{X}_{n} - \tilde{X}_{n-1} =1 \qquad \\  \label{eqn_coupling_above_2} \tilde{X}_{n} - \tilde{X}_{n-1} =-1 & \; \Rightarrow \; & Y_n = -1. \end{eqnarray} \end{enumerate}
\end{lemma}
\begin{proof} Property \textit{(a)} is a simple fact for any two $\{-1,0,1\}$ valued random variables $Z$ and $Z'$ with distributions say $\{d_1, d_2, d_3\}$ and $\{d_1', d_2', d_3'\}$. Assign $\mathbb{P}(Z=Z'=1):= \min\{d_3, d_3'\}$ and \textit{(a)} follows. To establish \textit{(b)} we analyse the dynamics of $\lancuch$ and consequently of $\lancucht.$ Recall Algorithm \ref{alg_Gibbs_adap} and the update rule for $\alpha_n$ in (\ref{eqn_formula_for_alpha}). Given $X_{n-1} = (i,j),$ the algorithm will obtain the value of $\alpha_n$ in step 1, next draw a coordinate according to $(\alpha_{n, 1}, \alpha_{n, 2})$ in step 2. In steps 3 and 4 it will move according to conditional distributions for updating the first or the second coordinate. These distributions are $$(1/2, 1/2) \qquad \textrm{and} \qquad \bigg(\frac{i^2}{i^2+(i-1)^2}, \frac{(i-1)^2}{i^2+(i-1)^2}\bigg)$$ respectively. Hence given $X_{n-1} = (i,i)$ the distribution of $X_{n} \in \{(i, i-1), (i,i), (i+1,i)\}$ is
\begin{equation}\label{eqn_X_dynamics_1}
\bigg( \big(\frac{1}{2}-\frac{4}{a_n}\big) \frac{i^2}{i^2+(i-1)^2}   \;,\;  1-\big(\frac{1}{2}-\frac{4}{a_n}\big) \frac{i^2}{i^2+(i-1)^2}  - \big( \frac{1}{4} + \frac{2}{a_n} \big)  \;,\; \frac{1}{4} + \frac{2}{a_n} \bigg),\;\;
\end{equation}
whereas if $X_{n-1} = (i,i-1)$ then $X_{n} \in \{(i-1, i-1), (i,i-1), (i,i)\}$ with probabilities
\begin{equation}\label{eqn_X_dynamics_2}
\bigg(\frac{1}{4} - \frac{2}{a_n}    \;,\;1 - \big(\frac{1}{4} - \frac{2}{a_n}  \big) -  \big(\frac{1}{2}+\frac{4}{a_n}\big)\frac{(i-1)^2}{i^2+(i-1)^2}  \;,\;  \big(\frac{1}{2}+\frac{4}{a_n}\big)\frac{(i-1)^2}{i^2+(i-1)^2} \bigg),\;\;
\end{equation}
respectively. We can conclude the evolution of $\lancucht.$ Namely, if $\tilde{X}_{n-1} = 2i-2$ then the distribution of $\tilde{X}_n - \tilde{X}_{n-1} \in \{-1,0,1\}$ is given by (\ref{eqn_X_dynamics_1}) and if $\tilde{X}_{n-1} = 2i-3$ then the distribution of $\tilde{X}_n - \tilde{X}_{n-1} \in \{-1,0,1\}$ is given by (\ref{eqn_X_dynamics_2}). Let $\leq_{\textrm{st}}$ denote stochastic ordering. By simple algebra both measures defined in (\ref{eqn_X_dynamics_1}) and (\ref{eqn_X_dynamics_2}) are stochastically bigger then
\begin{eqnarray}\label{eqn_X_dynamics_bound}
\mu_n^i & = & (\mu_{n,1}^i, \mu_{n,2}^i, \mu_{n,3}^i ), \end{eqnarray}
where
\begin{eqnarray}
\mu_{n,1}^i & = & \big(\frac{1}{4} - \frac{2}{a_n}\big)\big(1+\frac{2}{i}\big) \;\; = \;\; \frac{1}{4} - \frac{1}{a_n} - \frac{2i+8 - a_n}{2ia_n}, \label{eqn_intermeasure_minus_1} \\
\mu_{n,2}^i & = & 1 -\big(\frac{1}{4} - \frac{2}{a_n}\big)\big(1+\frac{2}{i}\big) - \big(\frac{1}{4}+\frac{2}{a_n}\big)\big(1-\frac{2}{\max\{4,i\}}\big),\nonumber \\\mu_{n,3}^i &=& \big(\frac{1}{4}+\frac{2}{a_n}\big)\big(1-\frac{2}{\max\{4,i\}}\big) \;\;=\;\; \frac{1}{4} + \frac{1}{a_n} + \frac{2\max\{4,i\}-8 - a_n}{2a_n \max\{4,i\}}.\label{eqn_intermeasure_plus_1} \qquad
\end{eqnarray}
Recall $\nu_n,$ the distribution of $Y_n$ defined in (\ref{eqn_distr_of_Y}). Examine (\ref{eqn_intermeasure_minus_1}) and (\ref{eqn_intermeasure_plus_1}) to see that if $2i - 8 \geq a_n,$ then $\mu_n^i \geq_{\textrm{st}} \nu_n.$ Hence in this case also the distribution of $\tilde{X}_{n} - \tilde{X}_{n-1}$ is stochastically bigger then the distribution of $Y_n.$ The joint probability distribution of $(\tilde{X}_{n} - \tilde{X}_{n-1}, Y_n)$ satisfying (\ref{eqn_coupling_above_1}) and (\ref{eqn_coupling_above_2}) follows. \end{proof}
\begin{proof}[Proof of Proposition~\ref{fact_X_nonergodic}] Define \begin{equation}
\Omega_{1, \tilde{X}}  := \Big\{\omega \in \Omega: \tilde{X}_n - \tilde{X}_{n-1} =1 \quad \textrm{for every} \quad 0 < n \leq c_1\Big\}.
\end{equation} Since the distribution of $\tilde{X}_n - \tilde{X}_{n-1}$ is stochastically bigger then $\mu_n^i$ defined in (\ref{eqn_X_dynamics_bound}) and $\mu_n^i(1) > c > 0$ for every $i$ and $n,$ \begin{equation}\nonumber
\mathbb{P}\big(\Omega_{1, \tilde{X}}\big) =: p_{1,\tilde{X}} > 0.
\end{equation}
By Lemma \ref{lemma_coupling_above} \textit{(a)} we have
\begin{equation}
\mathbb{P}\big(\Omega_{1, \tilde{X}}\cap \Omega_1 \big) \geq p_{1,S}\,p_{1,\tilde{X}} > 0.
\end{equation}
Since $S_{c_1}=\tilde{X}_{c_1} = c_1 = b_1,$ on $\Omega_{1, \tilde{X}}\cap \Omega_1,$ the requirements for Lemma \ref{lemma_coupling_above} \textit{(b)} hold for $n-1=c_1.$ We shall use Lemma \ref{lemma_coupling_above} \textit{(b)}  iteratively to keep $\tilde{X}_n \geq S_n$ for every $n.$ Recall that we write $\tilde{X}_{n-1}$ as $\tilde{X}_{n-1} = 2i-2$ or $\tilde{X}_{n-1} = 2i-3.$ If $2i - 8 \geq a_n$ and $\tilde{X}_{n-1} \geq S_{n-1}$ then by Lemma \ref{lemma_coupling_above} \textit{(b)} also $\tilde{X}_{n} \geq S_{n}.$
Clearly if $\tilde{X}_{k} \geq S_{k}$ and $S_k \geq \frac{b_{n-1}}{2}$ for $c_{n-1} < k \leq c_n$ then $\tilde{X}_k \geq \frac{b_{n-1}}{2}$ for $c_{n-1} < k \leq c_n,$ hence $$2i-2 \geq  \frac{b_{n-1}}{2} \quad \textrm{for}\quad  c_{n-1} < k \leq c_n.$$
This in turn gives $2i - 8 \geq \frac{b_{n-1}}{2} - 6$ for $c_{n-1} < k \leq c_n$ and since $a_k = 10 + \log(n),$ for the iterative construction to hold, we need $b_{n} \geq 32+2\log(n+1).$ By the definition of $b_n$ and standard algebra we have
$$b_n \geq 1000\left(1 + \sum_{i=2}^n \frac{1}{10+\log(n)}\right) \geq 32+2\log(n+1) \quad \textrm{for every } \; n\geq 1.$$
Summarising the above argument provides
\begin{eqnarray*}
\mathbb{P}(X_{n, 1} \to \infty) & \geq & \mathbb{P}\left(\Omega_{\infty} \cap \Omega_{\tilde{X}\geq S} \right) \;\; \geq \;\; \mathbb{P}\Bigg(\bigg(\bigcap_{n=1}^{\infty} \Omega_n\bigg) \cap \Omega_{\tilde{X}\geq S}\Bigg) \\ &\geq& \mathbb{P}\Bigg(\Omega_{1, \tilde{X}} \cap \bigg(\bigcap_{n=1}^{\infty} \Omega_n\bigg) \cap \Omega_{\tilde{X}\geq S}\Bigg) \\ & \geq & p_{1,\tilde{X}}p_{1,S} \prod_{n=2}^{\infty} (1-p_n) \;\; > \;\; 0.
\end{eqnarray*}
Hence $\lancuch$ is not ergodic, and in particular $\|\pi_n - \pi\|_{\textrm{TV}}  \nrightarrow 0.$
\end{proof}

\section*{Acknowledgements} This paper was written while the first author
was a postdoctoral fellow at the Department of Statistics, University
of Toronto.  Both authors were partially funded by NSERC of Canada.


\end{document}